\documentclass[11pt,a4paper]{article}
\usepackage{amsmath}
\usepackage{amsfonts}
\usepackage{amssymb}
\usepackage{makeidx}
\usepackage{graphicx}
 \usepackage[a4paper,left=3cm,right=3cm,top=2.5cm,bottom=2.5cm,includefoot,headheight=13.6pt]{geometry}
\usepackage{algorithm,algorithmic}

\newtheorem{theorem}{Theorem}

\newtheorem{definition}[theorem]{Definition}

\newtheorem{example}[theorem]{Example}
\newtheorem{remark}[theorem]{Remark}
\newenvironment{proof}[1][Proof]{\noindent\textbf{#1.} }{\ \rule{0.5em}{0.5em}}

\begin{document}

\author{}
\title{\textbf{Power divergence approach for one-shot device testing under competing risks. }}
\date{}
\author{N. Balakrishnan\footnote{McMaster University, ON, Canada. email: bala@mcmaster.ca }, E. Castilla\footnote{Complutense University of Madrid, Spain. email: elecasti@ucm.es }, N. Martin\footnote{Complutense University of Madrid, Spain. email: nirian@estad.ucm.es} \ and L. Pardo\footnote{Complutense University of Madrid, Spain. email: lpardo@ucm.es }}
\maketitle

\begin{abstract}
Most work on one-shot devices assume that there is only one possible cause of  device failure. However, in practice, it is often the case that  the products under study can experience any one of various possible causes of failure. Robust estimators and Wald-type tests  are developed here for the case of one-shot devices under competing risks.  An extensive simulation study illustrates the robustness of these divergence-based estimators and test procedures based on them. A data-driven procedure is  proposed for choosing the optimal estimator for any given data set which is then  applied to an example in the context of survival analysis.
\end{abstract}
\section{Introduction}

In lifetime data analysis, it is often the case that the products under study can experience  one of different types of failure. For example, in the context of survival analysis, we can have several different types of failure (death, relapse, opportunistic infection, etc.) that are of interest to us, leading to the so-called ``competing risks'' scenario.  A competing risk is an event whose occurrence precludes the occurrence of the primary event of interest. In a study examining time to death attributable, for instance, to cardiovascular causes, death attributable to noncardiovascular causes would be a competing risk. Crowder (2001) has presented review of this competing risks problem for which one needs to estimate the failure rates for each cause.  Balakrishnan  et al. (2016a, 2016b) and So (2016) have discussed the problem of one-shot devices under competing risk   for the first time. 

One-shot device testing data, also known as current status data in survival analysis, come from testing one-shot devices that are tested only once and get destroyed right after the test. It is very common that one-shot devices contain multiple components and failure of any of them will lead to the failure of the device. In Balakrishnan et al (2019a, 2019b, 2019c, 2020), new estimators and Wald-type tests for one-shot devices models have been proposed. The estimators introduced, namely, weighted minimum density power divergence estimators, and the corresponding Wald-type tests, demonstrated good behaviour and performance in terms of robustness without a significant loss of efficiency. However, it was assumed that there is only one survival endpoint of interest, and that censoring is independent of the event in interest. The main purpose of this work is to develop  weighted minimum density power divergence estimators as well as Wald-type test statistics under competing risk models for one-shot device testing assuming exponential lifetimes.

In Section  \ref{sec:model}, we  present the model formulation as well as the notation to be used the rest of the paper. In this section, we also describe the relation between  maximum likelihood estimator and   minimization of the Kullback-Leibler divergence between appropriate distributions.  The weighted minimum density power divergence estimators  for one-shot device testing exponential model  under competing risks are then developed in Section \ref{sec:MDPDE}.   Their asymptotic distribution and a new family of Wald-type test statistics based on them are also  presented in this section. In Section \ref{sec:sim},  an extensive Monte Carlo simulation study is carried out for demonstrating the robust behaviour of the proposed estimators as well as the testing procedures.  An ad hoc procedure for the choice of the optimal estimator  is also provided in this section. The developed methods are then applied to a pharmacology data for illustrative purposes. Finally, some concluding remarks are made in Section \ref{sec:conclusion}, while the proofs of all the main results are presented in Appendix \ref{sec:app}.

\section{Model formulation \label{sec:model}}
In this section, we shall introduce the notation necessary for the developments in this paper, paying special attention to the maximum likelihood estimator (MLE) of the model, as well as its relation with the minimization of  Kullback-Leibler divergence.  

The setting for an accelerate life-test for one-shot devices under competing risks considered here is stratified in $I$ testing conditions as follows:
\begin{enumerate}
\item The tests are checked at inspection times $IT_i$, for $i=1,\dots,I$;
\item The devices are tested under $J$ different stress levels, $\boldsymbol{x}_i=(x_{i1},\dots,x_{iJ})^T$, for $i=1,\dots,I$;
\item $K_i$ devices are tested in the $i$th test condition, for $i=1,\dots,I$;
\item The number of devices failed due to the $r$-th cause under the $i$-th test condition is denoted by $n_{ir}$, for $i=1,\dots I$, \ $r=1,\dots,R$;
\item The number of devices that survive under the $i$-th test condition is denoted by $n_{i0}=K_i-\sum_{r=1}^Rn_{ir}$.
\end{enumerate}

\begin{table}[htbp]  \tabcolsep4pt  \centering%
\caption{One-shot device testing under competing risks.\label{table:model}}
\begin{tabular}{c ccc ccc c ccc}
\hline
& & & & \multicolumn{3}{c}{Failures} &&\multicolumn{3}{c}{Stress levels}\\
\cline{5-7} \cline{9-11}
Condition & Times &  Devices & Survivals  & Cause $1$ &$\cdots $ & Cause $R$& & Stress $1$ &$\cdots $ & Stress $J$  \\  
\hline
$1$ & $IT_{1}$ & $K_{1}$ & $n_{10}$ &   $n_{11}$& $\cdots $&   $n_{1R}$& &  $x_{11}$ & $\cdots $ & $x_{1J}$ \\ 
$2$ & $IT_{2}$ & $K_{2}$ & $n_{20}$ & $n_{21}$& $\cdots $&   $n_{2R}$ & &$x_{21}$ & $\cdots $ & $x_{2J}$ \\ 
$\vdots $ & $\vdots $ & $\vdots $ &  & $\vdots $& & $\vdots $ &  & $\vdots $& & $\vdots $ \\ 
$I$ & $IT_{I}$ & $K_{I}$ & $n_{I0}$ &  $n_{I1}$& $\cdots $&   $n_{IR}$ && $x_{I1}$ & $\cdots $ & $x_{IJ}$ \\ 
\hline
\end{tabular}%
\end{table}%

This setting is summarized in Table \ref{table:model}. For simplicity, and as considered in Balakrishnan et al. (2016a), we will limit in this paper, the number of stress levels to $J=1$ and the number of competing causes to $R=2$, even though  inference for the general case when $J>1$ and $R>2$ can be presented in an analogous manner.

Let us denote the random variable for the failure time due to causes 1 and 2 as $T_{irk}$, for $r = 1,2$, $i=1,\dots,I$,  and $k = 1,\dots ,K _i$, respectively. We now assume that $T_{irk}$ follows an exponential distribution with failure rate parameter $\lambda_{ir}(\boldsymbol{\theta})$ and  its probability density function 
\begin{align*}
f_r(t;x_i,\boldsymbol{\theta})&=\lambda_{ir}(\boldsymbol{\theta})e^{-\lambda_{ir}(\boldsymbol{\theta})t}, \quad t>0,\\
 \lambda_{ir}(\boldsymbol{\theta})&=\theta_{r0}\exp({\theta_{r1}x_i}),\\\boldsymbol{\theta}&=(\theta_{10},\theta_{11},\theta_{20},\theta_{21})^T, \quad \theta_{r0},\theta_{r1}>0 , \quad r=1,2,
\end{align*}
where $x_i$ is the stress factor of the condition $i$ and $\boldsymbol{\theta}$ is the model parameter vector, with $\boldsymbol{\theta}\in\mathbb{R}^{4}$.

 We shall use $\pi_{i0}(\boldsymbol{\theta})$, $\pi_{i1}(\boldsymbol{\theta})$ and $\pi_{i2}(\boldsymbol{\theta})$ for the survival probability, failure probability due to cause $1$ and failure probability due to cause $2$, respectively. Their expressions are
\begin{align*}
\pi_{i0}(\boldsymbol{\theta})&=(1-F_1(IT_i;x_i,\boldsymbol{\theta}))(1-F_2(IT_i;x_i,\boldsymbol{\theta}))=\exp(-(\lambda_{i1}+\lambda_{i2})IT_i),\\
\pi_{i1}(\boldsymbol{\theta})&=\frac{\lambda_{i1}}{\lambda_{i1}+\lambda_{i2}}(1-\exp(-(\lambda_{i1}+\lambda_{i2})IT_i)),\\
\pi_{i2}(\boldsymbol{\theta})&=\frac{\lambda_{i2}}{\lambda_{i1}+\lambda_{i2}}(1-\exp(-(\lambda_{i1}+\lambda_{i2})IT_i)),
\end{align*}
where $\lambda_{ir}=\lambda_{ir}(\boldsymbol{\theta}), \ r=1,2$. Derivations of these expressions can be found in So (2016, p. 151). Now,  the likelihood function  is given by

\begin{equation}\label{eq:likelihood}
\mathcal{L}(n_{01},\dots,n_{I2};\boldsymbol{\theta})
\propto\prod_{i=1}^I \pi_{i0}(\boldsymbol{\theta})^{n_{i0}}\pi_{i1}(\boldsymbol{\theta})^{n_{i1}}\pi_{i2}(\boldsymbol{\theta})^{n_{i2}},
\end{equation}
where $n_{0i}+n_{1i}+n_{2i}=K_i, \ i=1,\dots,I$.

\begin{definition}[MLE, classical definition]
The maximum likelihood estimator (MLE) of $\boldsymbol{\theta}$, denoted by $\widehat{\boldsymbol{\theta}}$, is obtained by maximizing the likelihood function in (\ref{eq:likelihood}) or, equivalently, its logarithm.
\end{definition}

We will present an alternative definition of the MLE later on (see Definition \ref{def:MLE_Alt}). Let us introduce the following probability vectors:
\begin{align}
\widehat{\boldsymbol{p}}_{i}&=(\widehat{p}_{i0},\widehat{p}_{i1},\widehat{p}_{i2})^T=\frac{1}{K_i}(n_{i0},n_{i1},n_{i2})^T, \quad i=1,\dots,I, \label{eq:observed_vector}\\ 
\boldsymbol{\pi}_{i}(\boldsymbol{\theta})&=(\pi_{i0}(\boldsymbol{\theta}),\pi_{i1}(\boldsymbol{\theta}),\pi_{i2}(\boldsymbol{\theta}))^T, \quad i=1,\dots,I. \label{eq:theo_vector} 
\end{align}
The Kullback-Leibler divergence measure (see, for instance, Pardo (2006)), between $\widehat{\boldsymbol{p}}_{i}$ and $\boldsymbol{\pi }_{i}(\boldsymbol{\theta})$, is given by 
\begin{align*}
d_{KL}(\widehat{\boldsymbol{p}}_{i},\boldsymbol{\pi }_{i}(\boldsymbol{\theta}))& =\sum_{r=0}^2 \widehat{p}_{ir}\log \left( \dfrac{\widehat{p}_{ir}}{\pi _{ir}(\boldsymbol{\theta})}\right)\notag\\
&=\widehat{p}_{i0}\log \left( \dfrac{\widehat{p}_{i0}}{\pi _{i0}(\boldsymbol{\theta})}\right) +\widehat{p}_{i1}\log \left( \dfrac{\widehat{p}_{i1}}{\pi _{i1}(\boldsymbol{\theta})}\right)+\widehat{p}_{i2}\log \left( \dfrac{\widehat{p}_{i2}}{\pi _{i2}(\boldsymbol{\theta})}\right)  \notag \\
& =\frac{n_{i0}}{K_i}\log \left( \dfrac{n_{i0}/K_i}{\pi _{i0}(\boldsymbol{\theta})}\right) +\frac{n_{i1}}{K_i}\log \left( \dfrac{n_{i1}/K_i}{\pi _{i1}(\boldsymbol{\theta})}\right) +\frac{n_{i2}}{K_i}\log \left( \dfrac{n_{i2}/K_i}{\pi _{i2}(\boldsymbol{\theta})}\right) \notag\\
& = \frac{1}{K_i}\left\{ {n_{i0}}\log \left( \dfrac{n_{i0}/K_i}{\pi _{i0}(\boldsymbol{\theta})}\right) +{n_{i1}}\log \left( \dfrac{n_{i1}/K_i}{\pi _{i1}(\boldsymbol{\theta})}\right) +{n_{i2}}\log \left( \dfrac{n_{i2}/K_i}{\pi _{i2}(\boldsymbol{\theta})}\right)\right\},
 \label{eq:kull}
\end{align*}%
and the weighted Kullback-Leibler divergence measure is given by
\begin{align*}
d_{KL}^{W}(\boldsymbol{\theta})&=\sum_{i=1}^{I}\frac{K_{i}}{K}d_{KL}(\widehat{\boldsymbol{p}}_{i},\boldsymbol{\pi }_{i}(\boldsymbol{\theta}))\\
&=\frac{1}{K}\sum_{i=1}^{I}\left\{ {n_{i0}}\log \left( \frac{n_{i0}/K_i}{\pi _{i0}(\boldsymbol{\theta})}\right) 
+{n_{i1}}\log \left( \frac{n_{i1}/K_i}{\pi _{i1}(\boldsymbol{\theta})}\right) 
+{n_{i2}}\log \left( \frac{n_{i2}/K_i}{\pi _{i2}(\boldsymbol{\theta})}\right)\right\} ,
\end{align*}%
with $K=K_{1}+\cdots +K_{I}$.

\begin{theorem}
\label{res:dkull} The likelihood function $\mathcal{L}(n_{01},\dots,n_{I2};\boldsymbol{\theta})$, given in (\ref{eq:likelihood}), is
related to the weighted Kullback-Leibler divergence measure through 
\begin{equation}\label{eq:MLE_ALT}
d_{KL}^{W}(\boldsymbol{\theta})=\sum_{i=1}^{I}\frac{K_{i}}{K}d_{KL}(\widehat{\boldsymbol{p}}_{i},\boldsymbol{\pi }_{i}(\boldsymbol{\theta}))=c-\frac{1}{K}\log \mathcal{L}(n_{01},\dots,n_{I2};\boldsymbol{\theta}),
\end{equation}%
with $c$ being a constant, not dependent on $\boldsymbol{\theta}$.
\end{theorem}

\begin{definition}[MLE, alternative definition] \label{def:MLE_Alt}
The MLE of $\boldsymbol{\theta}$, $\widehat{\boldsymbol{\theta}}$, can be obtained by the minimization of the weighted Kullback-Leibler divergence measure given in (\ref{eq:MLE_ALT}).
\end{definition}

\begin{example}[The BDC experiment]
 The  benzidine dihydrochloride (BDC)  experiment, studied in  Lindsey and Ryan (1993) and conducted at the National Center for Toxicological Research, examines the incidence in mice of liver tumors induced by the drug. Two different  doses of drug are induced in the mice: 60 parts per million (w=1) and 400 parts per million (w=2) and two causes of death are recorded: died without tumor ($\delta_{ijk}=1$) and died with tumor ($\delta_{ijk}=2$). These data are presented in Table \ref{table:numdata}.
\begin{table}[h!]
\caption{BDC experiment \label{table:numdata}}
\centering
\begin{tabular}{ l l ccc}
\hline
  &  & $\delta_{ijk}=0$ & $\delta_{ijk}=1$ & $\delta_{ijk}=2$ \\
\hline
$IT_1=9.37$  & $w=1$ & 70 & 2 & 0 \\
 & $w=2$ &  22 & 3 & 0 \\ \hline
$IT_2=14.07$  & $w=1$ & 48 & 1 & 0 \\
 & $w=2$ &  14 & 4 & 17 \\ \hline
 $IT_3=18.7$  & $w=1$ & 35 & 4 & 7 \\
 & $w=2$ &  1 & 1 & 9 \\ \hline
\end{tabular}
\end{table}
With these data, we obtained the MLE of the vector of parameters and also measured the  discrepancy of the corresponding estimated rates and the observed ones, given by
\begin{equation}\label{eq:estimated_error}
\frac{1}{3I}\sum_{i=1}^{I}\sum_{r=0}^2\left|\frac{n_{ir}-K_i\pi_{ir}(\widehat{\boldsymbol{\theta}})}{K_i}\right|.
\end{equation}
These results are presented in Table \ref{table:numdata0}. In the ensuing work, we will present alternative estimators to the MLE, which are seen to provide better performance in terms of robustness.

\begin{table}[h!]
\centering
\caption{MLEs of parameters based  on the BDC experiment \label{table:numdata0}}
\small
\begin{tabular}{l cccc c}
\hline
&$\theta_{10}$ & $\theta_{11}$ & $\theta_{20}$ & $\theta_{21}$ & estimated error  \\  \hline
MLE  & 0.00089   & 1.3191   & 0.00028   & 2.493 & 0.1051  \\
\hline
\end{tabular}
\end{table}

\end{example}

\section{Weighted minimum density power divergence estimator\label{sec:MDPDE}}

In this section, we shall introduce the weighted minimum density power divergence estimator  as a natural extension of the MLE. For this purpose, we shall introduce  the ordinary density power divergence (DPD). Given these two probability vectors $\widehat{\boldsymbol{p}}_{i}$ and $\boldsymbol{\pi }_{i}(\boldsymbol{\theta})$, defined in (\ref{eq:observed_vector}) and (\ref{eq:theo_vector}), respectively, the DPD between both probability vectors is given by 
\begin{align*}
d_{\beta }(\widehat{\boldsymbol{p}}_{i},\boldsymbol{\pi }_{i}(\boldsymbol{\theta}%
))& =\left( \pi _{i0}^{\beta +1}(\boldsymbol{\theta})+\pi _{i1}^{\beta +1}(\boldsymbol{\theta})+\pi _{i2}^{\beta +1}(%
\boldsymbol{\theta})\right) \notag \\
& -\frac{\beta +1}{\beta }\left( \widehat{p}_{i0}\pi
_{i0}^{\beta }(\boldsymbol{\theta})+\widehat{p}_{i1}\pi
_{i1}^{\beta }(\boldsymbol{\theta})+\widehat{p}_{i2}\pi _{i2}^{\beta }(%
\boldsymbol{\theta})\right)  \notag \\
& +\frac{1}{\beta }\left(\widehat{p}_{i0}^{\beta +1}+ \widehat{p}_{i1}^{\beta +1}+\widehat{p}%
_{i2}^{\beta +1}\right) ,\quad \text{if }\beta >0,  
\end{align*}%
and $d_{\beta =0}(\widehat{\boldsymbol{p}}_{i},\boldsymbol{\pi }_{i}(\boldsymbol{\theta}))=\lim_{\beta \rightarrow 0^{+}}d_{\beta }(\widehat{\boldsymbol{p}}_{i},\boldsymbol{\pi }_{i}(\boldsymbol{\theta}))=d_{KL}(\widehat{\boldsymbol{p}}_{i},\boldsymbol{\pi }_{i}(\boldsymbol{\theta}))$, for $\beta =0$.

The weighted DPD is given by 
\begin{align*}
d_{\beta}^{W}(\boldsymbol{\theta})=&\sum_{i=1}^{I}\frac{K_{i}}{K}\left[
\left(\pi _{i0}^{\beta +1}(\boldsymbol{\theta})+ \pi _{i1}^{\beta +1}(\boldsymbol{\theta})+\pi _{i2}^{\beta +1}(\boldsymbol{\theta})\right) \right.\\
&\left. -\frac{\beta +1}{\beta }\left( \widehat{p}_{i0}\pi _{i0}^{\beta }(%
\boldsymbol{\theta})+\widehat{p}_{i1}\pi _{i1}^{\beta }(\boldsymbol{\theta})+\widehat{p}_{i2}\pi _{i2}^{\beta }(\boldsymbol{\theta})\right) +\frac{1}{\beta }\left( \widehat{p}_{i0}^{\beta +1}+\widehat{p}_{i1}^{\beta +1}+\widehat{p}_{i2}^{\beta
+1}\right) \right]
\end{align*}
but the term $\frac{1}{\beta }\left( \widehat{p}_{i0}^{\beta +1}+\widehat{p}_{i1}^{\beta +1}+\widehat{p}_{i2}^{\beta +1}\right) $, $i=1,...,I$, does not have any role in its minimization with respect to $\boldsymbol{\theta}$. Therefore, in order to minimize $d_{\beta}^{W}(\boldsymbol{\theta})$, we can consider the equivalent measure 

\begin{small}
\begin{align}
^*d_{\beta}^{W}(\boldsymbol{\theta})=
\sum_{i=1}^I\frac{K_i}{K}&\left[\left(\pi _{i0}^{\beta +1}(\boldsymbol{\theta})+ \pi _{i1}^{\beta +1}(\boldsymbol{\theta})+\pi _{i2}^{\beta +1}(\boldsymbol{\theta})\right) \right. \notag \\
& \left. -\frac{\beta +1}{\beta }\left( \widehat{p}_{i0}\pi _{i0}^{\beta }(\boldsymbol{\theta})+\widehat{p}_{i1}\pi _{i1}^{\beta }(\boldsymbol{\theta})+\widehat{p}_{i2}\pi _{i2}^{\beta }(\boldsymbol{\theta})\right)\right].
\label{eq:dpd}
\end{align}
\end{small}

\begin{definition}
We can define the  weighted minimum density power divergence estimator of $\boldsymbol{\theta}$ as
\begin{equation*}
\widehat{\boldsymbol{\theta}}_{\beta }=\underset{\boldsymbol{\theta}\in \Theta }{\arg
\min }^*d_{\beta}^{W}(\boldsymbol{\theta}),\quad \text{for }\beta >0
\end{equation*}%
and for $\beta =0$ we get the weighted maximum likelihood estimator.
\end{definition}

\begin{theorem}
\label{res:est_equations2} The weighted minimum density power estimator of $\boldsymbol{\theta}$, with tuning parameter $\beta\geq 0$, $\widehat{\boldsymbol{\theta}}_{\beta}$, can be obtained as the solution of the following equation:
\begin{small}
\begin{equation*}
\sum_{i=1}^{I}K_{i}\left\{-\pi_{i0}(\boldsymbol{\theta})IT_i\left[ \pi_{i0}(\boldsymbol{\theta})^{\beta-1}(\pi_{i0}(\boldsymbol{\theta})-p_{i0})-(1-\pi_{i0}(\boldsymbol{\theta}))^{\beta-1}\Gamma_{i,\beta}\right]\boldsymbol{l}_i+(1-\pi_{i0}(\boldsymbol{\theta}))^{\beta}\Gamma^*_{i,\beta}\right\} =\boldsymbol{0}_{4},
\end{equation*}%
\end{small}
where 

\begin{align*}
\Gamma_{i,\beta}&=\dfrac{\lambda_{i1}^{\beta}\left[\frac{\lambda_{i1}}{\lambda_{i1}+\lambda_{i2}}(1-\pi_{i0}(\boldsymbol{\theta}))-p_{i1} \right]+\lambda_{i2}^{\beta}\left[\frac{\lambda_{i2}}{\lambda_{i1}+\lambda_{i2}}(1-\pi_{i0}(\boldsymbol{\theta}))-p_{i2} \right]}{(\lambda_{i1}+\lambda_{i2})^{\beta}},\\
\Gamma^*_{i,\beta}&=\dfrac{\lambda_{i1}^{\beta-1}\left[\frac{\lambda_{i1}}{\lambda_{i1}+\lambda_{i2}}(1-\pi_{i0}(\boldsymbol{\theta}))-p_{i1} \right]-\lambda_{i2}^{\beta-1}\left[\frac{\lambda_{i2}}{\lambda_{i1}+\lambda_{i2}}(1-\pi_{i0}(\boldsymbol{\theta}))-p_{i2} \right]}{(\lambda_{i1}+\lambda_{i2})^{\beta-1}},
\end{align*}
$
\boldsymbol{l}_i=(\lambda_{i1}/\theta_{10},\lambda_{i1}x_i,\lambda_{i2}/\theta_{20},\lambda_{i2}x_i)^T$ and 
$\boldsymbol{r}_i=\frac{\lambda_{i1}\lambda_{i2}}{(\lambda_{i1}+\lambda_{i2})^2}(1/\theta_{10},x_i,-1/\theta_{20},-x_i)^T.
$

\end{theorem}

Now, by using the  Theorem 3.1 in Ghosh and Basu (2013), we can obtain the asymptotic distribution of the above weighted minimum density power divergence estimator.
\begin{theorem}\label{res:asymp}
 Let $\boldsymbol{\theta}^0$ be the true value of the parameter $\boldsymbol{\theta}$. The asymptotic distribution of the weighted minimum density power divergence estimator of $\boldsymbol{\theta}$, $\widehat{\boldsymbol{\theta}}_{\beta}$, is given by
\[
\sqrt{K}\left(  \widehat{\boldsymbol{\theta}}_{\beta}-\boldsymbol{\theta}^{0}\right)   \overset{\mathcal{L}}{\underset{K\mathcal{\rightarrow}\infty}{\longrightarrow}}\mathcal{N}\left(  \boldsymbol{0}_{4},\boldsymbol{{J}}_{\beta}^{-1}(\boldsymbol{\theta}^{0})\boldsymbol{{K}}_{\beta}(\boldsymbol{\theta}^{0})\boldsymbol{{J}}_{\beta}^{-1}(\boldsymbol{\theta}^{0})\right)  ,
\]
where
\begin{align}
\boldsymbol{{J}}_{\beta}(\boldsymbol{\theta}) &  =\sum_{i=1}^{I}\sum_{r=0}^{2}\frac{K_i}{K}\boldsymbol{u}^*_{ir}(\boldsymbol{\theta})\boldsymbol{u}_{ir}^{*T}(\boldsymbol{\theta})\pi_{ir}^{\beta-1}(\boldsymbol{\theta}), \label{eq:Jbar}\\
\boldsymbol{{K}}_{\beta}(\boldsymbol{\theta}) &  = \sum_{i=1}^{I}\sum_{r=0}^{2}\frac{K_i}{K}\boldsymbol{u}^*_{ir}(\boldsymbol{\theta})\boldsymbol{u}_{ir}^{*T}(\boldsymbol{\theta})\pi_{ir}^{2\beta-1}(\boldsymbol{\theta})-\sum_{i=1}^{I}\frac{K_i}{K}\boldsymbol{\xi}_{i,\beta}(\boldsymbol{\theta})\boldsymbol{\xi}_{i,\beta}^{T}(\boldsymbol{\theta}) \label{eq:Kbar},
\end{align}
with $ \boldsymbol{\xi}_{i,\beta}(\boldsymbol{\theta})   =\sum_{r=0}^{2}\boldsymbol{u}^*_{ir}(\boldsymbol{\theta})\pi_{ir}^{\beta}(\boldsymbol{\theta})$  and  $\boldsymbol{u}^*_{ir}(\boldsymbol{\theta})= \frac{\partial \pi_{ir}(\boldsymbol{\theta})}{\partial \boldsymbol{\theta}^T}$, where
\begin{align*}
\dfrac{\partial \pi_{i0}(\boldsymbol{\theta})}{\partial \boldsymbol{\theta}}&=-IT_i\pi_{i0}(\boldsymbol{\theta})\boldsymbol{l}_i, \\
\dfrac{\partial \pi_{i1}(\boldsymbol{\theta})}{\partial \boldsymbol{\theta}}&=\frac{\lambda_{i1}}{\lambda_{i1}+\lambda_{i2}}IT_i\pi_{i0}(\boldsymbol{\theta})\boldsymbol{l}_i+(1-\pi_{i0}(\boldsymbol{\theta}))\boldsymbol{r}_i,\\
\dfrac{\partial \pi_{i2}(\boldsymbol{\theta})}{\partial \boldsymbol{\theta}}&=\frac{\lambda_{i2}}{\lambda_{i1}+\lambda_{i2}}IT_i\pi_{i0}(\boldsymbol{\theta})\boldsymbol{l}_i-(1-\pi_{i0}(\boldsymbol{\theta}))\boldsymbol{r}_i,
\end{align*}
$
\boldsymbol{l}_i=(\lambda_{i1}/\theta_{10},\lambda_{i1}x_i,\lambda_{i2}/\theta_{20},\lambda_{i2}x_i)^T$ and 
$\boldsymbol{r}_i=\frac{\lambda_{i1}\lambda_{i2}}{(\lambda_{i1}+\lambda_{i2})^2}(1/\theta_{10},x_i,-1/\theta_{20},-x_i)^T.
$

\end{theorem}

\subsection{Wald-type test statistics}
Let us consider the function $\boldsymbol{m}:\mathbb{R}^{J+1}\longrightarrow 
\mathbb{R}^{r}$, where $r\leq 4$. Then 
\begin{equation}\label{eq:WALDtest}
\boldsymbol{m}\left( \boldsymbol{\theta}\right) =\boldsymbol{0}_{r},
\end{equation}%
with $\boldsymbol{0}_{r}$ being the null column vector of dimension $r$, which represents the  null hypothesis. We assume that the $4 \times r$ matrix  
\begin{equation*}
\boldsymbol{M}\left( \boldsymbol{\theta}\right) =\frac{\partial \boldsymbol{m}%
^{T}\left( \boldsymbol{\theta}\right) }{\partial \boldsymbol{\theta}}
\end{equation*}%
exists and is continuous in \textquotedblleft $\boldsymbol{\theta}$%
\textquotedblright\ and that rank$(\boldsymbol{M}\left( \boldsymbol{\theta}\right) )=r.$
For testing 
\begin{equation}
H_{0}:\boldsymbol{\theta\in \Theta }_{0}\text{ against }H_{1}:\boldsymbol{\theta\notin
\Theta }_{0},  \label{eq:HComp}
\end{equation}%
where 
\begin{equation*}
\boldsymbol{\Theta }_{0}=\left\{ \boldsymbol{\theta\in \Theta }_{0}:\boldsymbol{m}%
\left( \boldsymbol{\theta}\right) =\boldsymbol{0}_{r}\right\} ,
\end{equation*}%
we can consider the following Wald-type test statistics:
\begin{equation*}
W_{K}\left( \widehat{\boldsymbol{\theta}}_{\beta }\right) =K\boldsymbol{m}%
^{T}\left( \widehat{\boldsymbol{\theta}}_{\beta }\right) \left( \boldsymbol{M}%
^{T}\left( \widehat{\boldsymbol{\theta}}_{\beta }\right) \boldsymbol{\Sigma }%
\left( \widehat{\boldsymbol{\theta}}_{\beta }\right) \boldsymbol{M}\left( 
\widehat{\boldsymbol{\theta}}_{\beta }\right) \right) ^{-1}\boldsymbol{m}\left( 
\widehat{\boldsymbol{\theta}}_{\beta }\right) ,
\end{equation*}%
where%
\begin{equation*}
\boldsymbol{\Sigma }_{\beta }\left( \widehat{\boldsymbol{\theta}}_{\beta }\right)
={\boldsymbol{J}}_{\beta }^{-1}\left( \widehat{\boldsymbol{\theta}}_{\beta }\right) {\boldsymbol{K}}_{\beta }\left( \widehat{\boldsymbol{\theta}}_{\beta }\right) {\boldsymbol{J}}_{\beta}^{-1}\left( \widehat{\boldsymbol{\theta}}_{\beta }\right) ,
\end{equation*}%
and ${\boldsymbol{J}}_{\beta }\left( \boldsymbol{\theta}\right) $ and ${\boldsymbol{K}}_{\beta }\left( \boldsymbol{\theta}\right) $\ are
as given in (\ref{eq:Jbar}) and (\ref{eq:Kbar}), repectively. Wald-type test statistics based on weighted minimum density power divergence estimator
have been considered previously by Basu et al. (2015) and Ghosh et al. (2016).

\begin{theorem}
\label{th:test_asym} Under the null hypothesis, we have 
\begin{equation*}
W_{K}\left( \widehat{\boldsymbol{\theta}}_{\beta }\right) \underset{K\rightarrow
\infty }{\overset{\mathcal{L}}{\longrightarrow }}\chi _{r}^{2}.
\end{equation*}
\end{theorem}
Based on Theorem \ref{th:test_asym} , we can reject the null hypothesis,  in (\ref{eq:HComp}), if 
\begin{equation}\label{eq:reject}
W_{K}\left( \widehat{\boldsymbol{\theta}}_{\beta }\right) >\chi _{r,\alpha }^{2},
\end{equation}
where $\chi _{r,\alpha }^{2}$ is the upper $\alpha$ percentage point  of $\chi _{r}^{2}$ distribution.

Some results about the power function of the proposed Wald-type tests are presented in Appendix \ref{app:power}.\\

\begin{remark}[Robustness properties]
The influence function (IF) is a classical tool to measure the local robustness of an estimator (Hampel et al.,1968). In Balakrishnan et al. (2019a, 2020), the robustness of the weighted minimum density estimators and tests, for $\beta > 0$, was theoretically derived through loal dependence under the exponential assumption but in a non-competing risk framework, for large leverages $\boldsymbol{x}_is$. Analogous computations would result in the same conclusion for the competing risks scenario.   However, we could not directly infer about the robustness against outliers in the response variable which are, in fact, the misspecification errors. In the next section, a simulation study is carried out in order to empirically illustrate the  robustness of the proposed statistics with $\beta>0$, and the non-robustness when $\beta=0$, also against such misspecification errors.
\end{remark}



\section{Simulation Study \label{sec:sim}}
In this section, a Monte Carlo simulation study that examines the accuracy of the proposed weighted minimum density power divergence estimators is presented. Section \ref{sec:sim_eff} focuses on the efficiency, measured in terms of  root of mean square error (RMSE), mean  bias  error (MBE) and mean absolute error (MAE), of the estimators of model parameters, while Section \ref{sec:sim_test} examines the behavior of  Wald-type tests developed in preceding sections.  Finally, in Section \ref{sec:sim_choice}, an ad hoc procedure for the choice of the tuning parameter  is proposed. Every step of simulation was tested under S = 5,000 replications with  R statistical software.

Paying special attention to the robustness issue, we will consider in this context, ``outlying cells'' rather than ``outlying observations'' (see Balakrishnan et al., 2019a, 2019b). This means that devices under a specific testing condition (cell) will not follow the general one-shot device model considered, contributing to an increase in the values of the divergence measure between the data and the corresponding fitted values.
In this cell, the number of devices failed will be lower or higher than expected. This is similar to the principle of inflated models in distribution theory (see Lambert (1992) and Heilbron (1994)). The main purpose of this study is to show that within the family of weighted minimum density power divergence estimators, developed in the preceding sections, there are estimators with better robustness properties than the MLE, and the Wald-type tests constructed based on them are at the same time more robust than the classical Wald test constructed based on the MLE.

\subsection{Weighted minimum density power divergence estimators \label{sec:sim_eff}}
The lifetimes of devices are simulated  for different levels of reliability and different sample sizes, under 4 different stress conditions with 1 stress factor at 4 levels. Then, all devices under each stress condition are inspected at 3 different inspection times, depending on the level of reliability. The corresponding data will then be collected under $I=12$ test conditions.
\subsubsection{Balanced data: Effect of the sample size}
Firstly, a balanced data with equal sample size for each group was considered. $K_i$ was taken to range from small to large sample sizes, two causes of failure were considered, and the model parameters were set to be $\boldsymbol{\theta}=(\theta_{10},0.05,\theta_{20},0.08)^T$ with $\theta_{10}\in\{0.008,0.004,0.001\}$ and $\theta_{20}\in\{0.0008,0.0004,0.0001\}$  for devices with  low, moderate and high reliability, respectively.  To prevent many zero-observations in test groups, the inspection times were set as $IT \in \{5,10,20\}$ for the case of low reliability, $IT \in \{7,15,25\}$ for the case of moderate reliability, and $IT \in \{10,20,30\}$ for the case of high reliability. To evaluate the robustness of the weighted minimum density power divergence estimators, we studied their behavior in the presence of an outlying cell for the first testing condition in our table. This cell was generated under the parameters $\tilde{\boldsymbol{\theta}}=(\theta_{10},0.05,\theta_{20},0.15)^T$. See Table  \ref{table:modelsim1} for a summary of these scenarios. RMSEs, MAEs and MBEs of model parameters were then computed for the cases of both pure and contaminated data and are plotted in Figures \ref{fig:RMSE}, \ref{fig:MAE} and \ref{fig:MBE}, respectively, with similar conclusions for the three error measures.

For the case of pure data, MLE presents the best behaviour (overall in the model with high reliability) and an increment in the tuning parameter $\beta$ leads to a gradual loss in terms of efficiency. However, in the case of contaminated data, MLE turns to be the worst estimator, and weighted minimum density power divergence estimators with $\beta>0$ present  much more robust behaviour. Note that, as expected, an increase in the sample size improves the efficiency of the estimators, both for pure and contaminated data.

\begin{table}[h!!!]
\small
\caption{Parameter values used in the simulation. Study of efficiency. \label{table:modelsim1} }
\centering
\begin{tabular}{l l c c}
\hline
\textbf{Reliability }& \textbf{Parameters} & \textbf{Symbols} & \textbf{Values} \\ 
\hline
 &Risk 1   & $\theta_{10}$, $\theta_{11}$ & $0.008,0.05$ \\ 
\textit{Low reliability}&Risk 2   & $\theta_{20}$, $\theta_{21}$ & $0.0008,0.08$ \\ 
& Contamination & $\tilde{\theta}_{21}$ & $0.15$\\ 
&Temperature ($^\circ$C)   & $\boldsymbol{x}^T=(x_1,x_2,x_3,x_4)$ & $(35,45,55,65)$ \\ 
&Inspection Time (days) & $IT=\{IT_1,IT_2,IT_3\}$ & $\{5, 10, 20\}$ \\ 
\hline
&Risk 1   & $\theta_{10}$, $\theta_{11}$ & $0.004,0.05$ \\ 
\textit{Moderate reliability} &Risk 2   & $\theta_{20}$, $\theta_{21}$ & $0.0004,0.08$ \\ 
& Contamination & $\tilde{\theta}_{21}$ & $0.15$\\ 
&Temperature ($^\circ$C)   & $\boldsymbol{x}^T=(x_1,x_2,x_3,x_4)$ & $(35,45,55,65)$ \\ 
&Inspection Time (days) & $IT=(IT_1,IT_2,IT_3)$ & $\{7, 15, 25\}$ \\ 
\hline 
&Risk 1   & $\theta_{10}$, $\theta_{11}$ & $0.001,0.05$ \\ 
\textit{High reliability} &Risk 2   & $\theta_{20}$, $\theta_{21}$ & $0.0001,0.08$ \\
& Contamination & $\tilde{\theta}_{21}$ & $0.15$\\  
&Temperature ($^\circ$C)   & $\boldsymbol{x}^T=(x_1,x_2,x_3,x_4)$ & $(35,45,55,65)$ \\ 
&Inspection Time (days) & $IT=(IT_1,IT_2,IT_3)$ & $\{10, 20, 30\}$ \\ 
\hline 
\end{tabular} 
\end{table}

\subsubsection{Unbalanced data: Effect of the degree of contamination}
Now, we consider an unbalanced data with unequal sample sizes for the test conditions. This data set, which consists a total of $K=300$ devices, is presented in Table  \ref{table:ALTnb}.  A competing risks model, with two different causes of failure, was generated with parameters $\boldsymbol{\theta}= (0.001,0.05, 0.0001,0.08)^T$. To examine the robustness in this accelerated life test (ALT) plan (in which the devices are tested under high stress levels, so that more failures can be observed), we increased each  of the parameters of the outlying first cell (Figure \ref{fig:unbalanced}). The contaminated parameters are expressed by $\tilde{\theta}_{10},\tilde{\theta}_{11},\tilde{\theta}_{20}$ and $\tilde{\alpha}_{21}$, respectively.

\begin{table}[!!h!!!]  
\caption{ALT plan, unbalanced data.\label{table:ALTnb}}
\centering
\begin{tabular}{c|ccc}
\hline 
i & $x_i$ & $IT_i$ & $K_i$ \\  %
\hline 
1 & 35 & 10 & 50 \\ 
2 & 45 & 10 & 40 \\ 
3 & 55 & 10 & 20 \\ 
4 & 65 & 10 & 40 \\ 
5 & 35 & 20 & 20 \\ 
6 & 45 & 20 & 20 \\ 
7 & 55 & 20 & 30 \\ 
8 & 65 & 20 & 20 \\ 
9 & 35 & 30 & 20 \\ 
10 & 45 & 30 & 20 \\ 
11 & 55 & 30 & 10 \\ 
12 & 65 & 30 & 10 \\
\hline 
\end{tabular} 
\end{table}

When there is no contamination in the cell or the degree of contamination is very low, and in concordance with results obtained in the previous scenario,  MLE is observed to be the most efficient estimator. However, when  the degree of contamination increases, there is an increase in the error for all the estimators,  but weighted minimum density power divergence estimators are shown to be much more robust. This is also the case for whatever choice of the contamination parameters we considered.

\begin{figure}[h!]
\centering
\begin{tabular}{cc}
\includegraphics[scale=0.4]{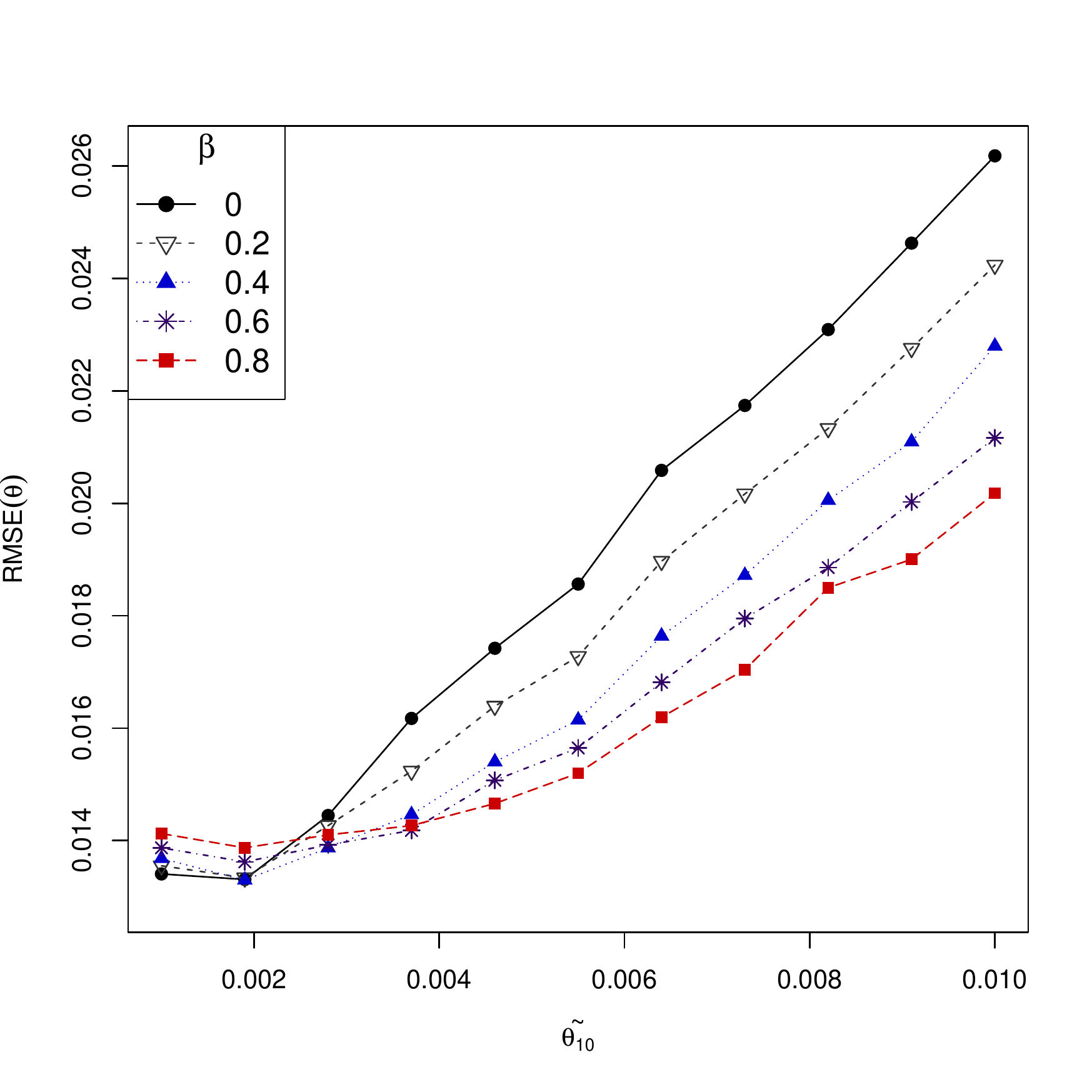} & 
\includegraphics[scale=0.4]{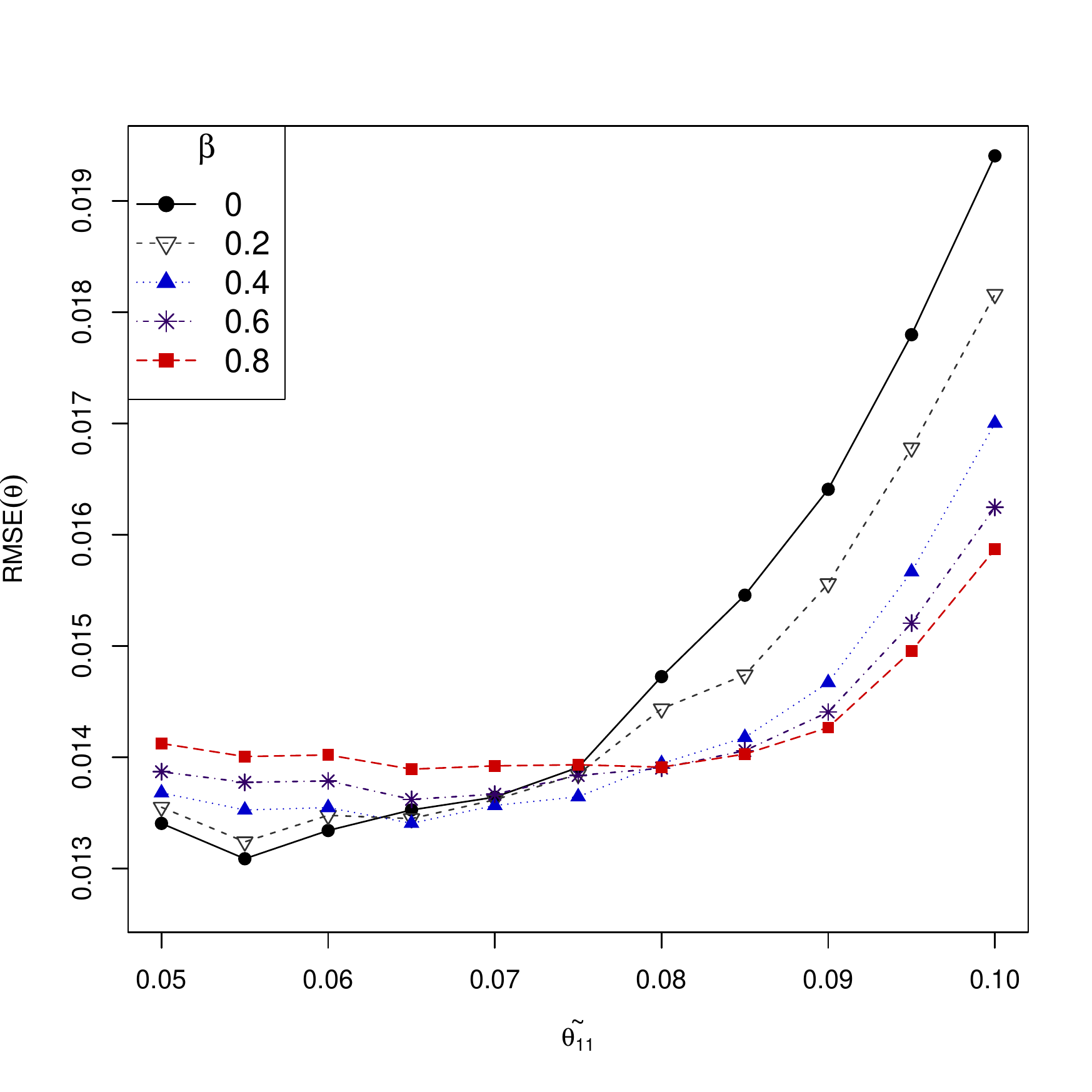} \\ 
\includegraphics[scale=0.4]{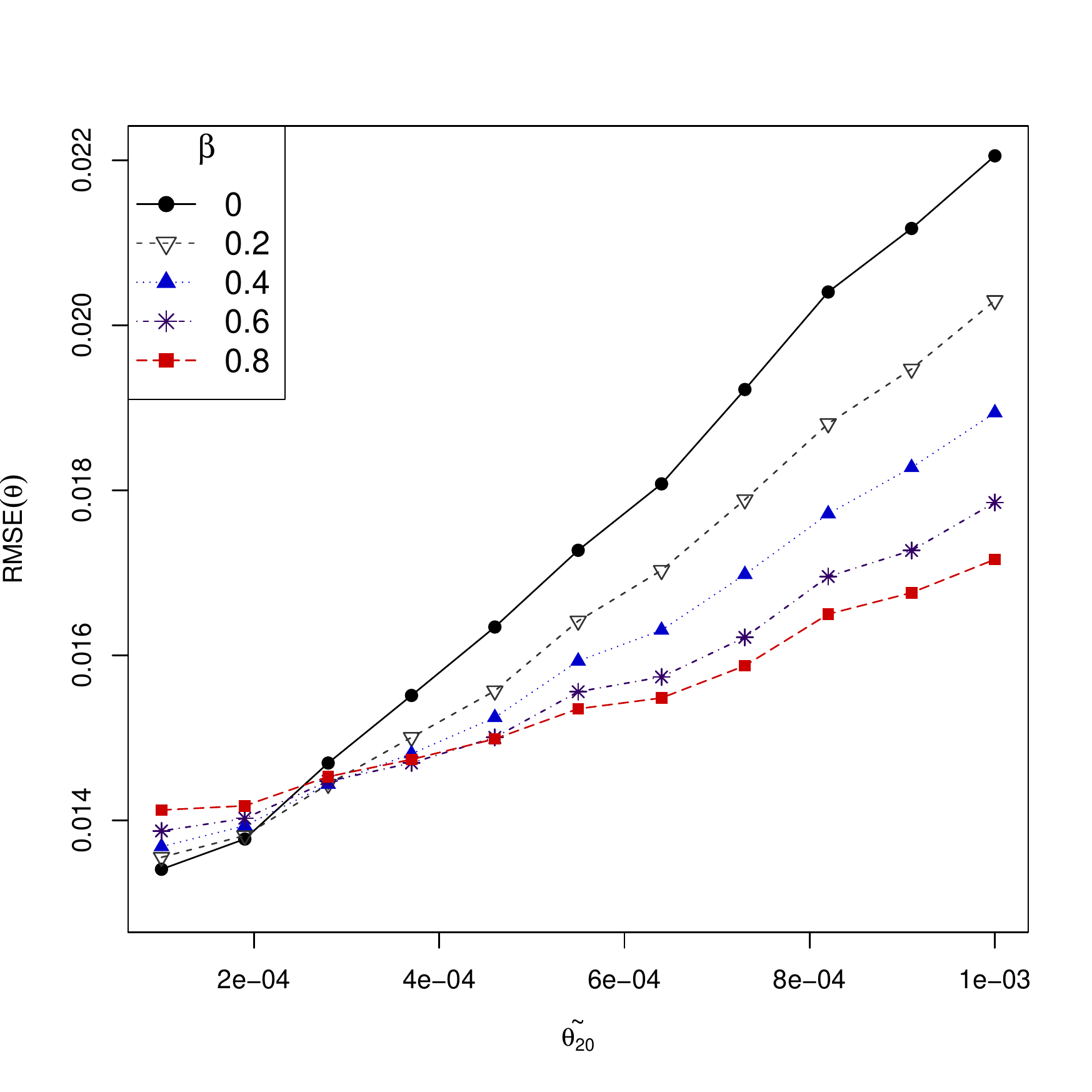} & 
\includegraphics[scale=0.4]{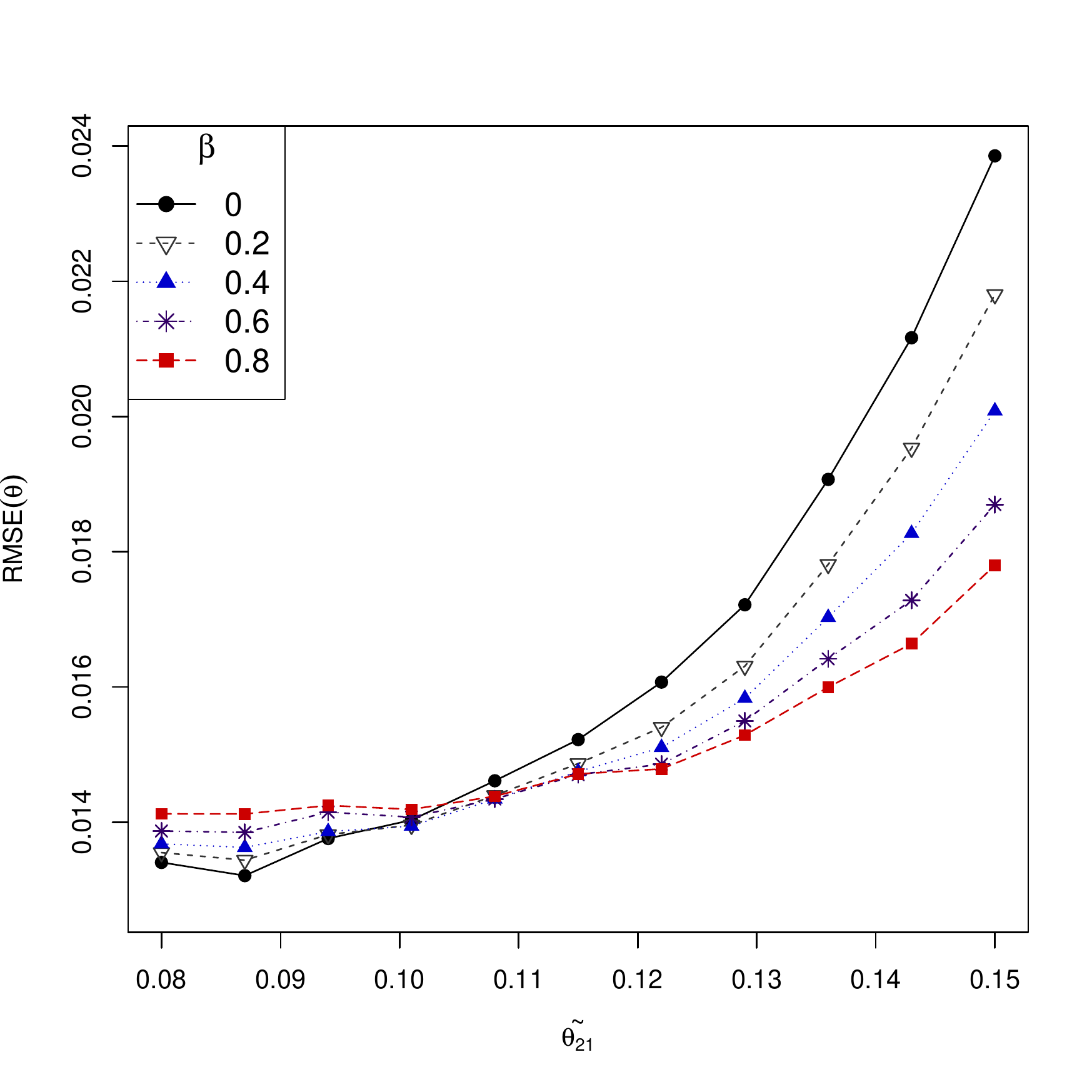} \\ 
\end{tabular}%
\caption{ RMSEs of the weighted minimum density power divergence estimators of $\boldsymbol{\theta}$ for different contamination parameter values. Unbalanced data.}
\label{fig:unbalanced}
\end{figure}

\subsection{Wald-type tests \label{sec:sim_test}}

Let us consider the balanced data under moderate reliability defined in the previous section. To compute the accuracy in terms of contrast, we  consider the  testing problem
\begin{equation}\label{eq:sim_hyp}
H_0:\theta_{21}=0.08 \quad \text{vs.} \quad H_1:\theta_{21}\neq 0.08.
\end{equation}
For computing the empirical test level, we measured the proportion of test statistics exceeding the corresponding chi-square critical value. The simulated test powers were also obtained under $H_1$ in (\ref{eq:sim_hyp}) in a similar manner. We used a nominal level of $0.05$. Table \ref{table:modelsimW} summarizes the model  considered for this purpose. As in the previous section, an outlying cell with $\tilde{\theta}_{21}=0.15$ is considered to illustrate the robustness of the proposed Wald-type tests (Figure \ref{fig:Wald}).

\begin{table}[h!!!!!] 
\caption{Parameter values used in the simulation study of Wald-type tests.\label{table:modelsimW}}
\centering
\begin{tabular}{l l c c}
\hline 
\textbf{Study} &  \textbf{Parameters} & \textbf{Symbols} & \textbf{Values} \\ 
\hline
\textit{Levels} & Model True Parameters  & $\boldsymbol{\theta}^T=(\theta_{10}, \theta_{11}, \theta{20}, \theta_{21})$ & $(0.004,0.05,0.0004,0.08)$ \\ 
\hline
\textit{Powers} & Model True Parameters  & $\boldsymbol{\theta}^T=(\theta_{10}, \theta_{11}, \theta{20}, \theta_{21})$ & $(0.004,0.05,0.0004,0.09)$ \\  

\hline  
\end{tabular} 
\end{table}

\begin{figure}[h!]
\centering
\begin{tabular}{ll}
\includegraphics[scale=0.4]{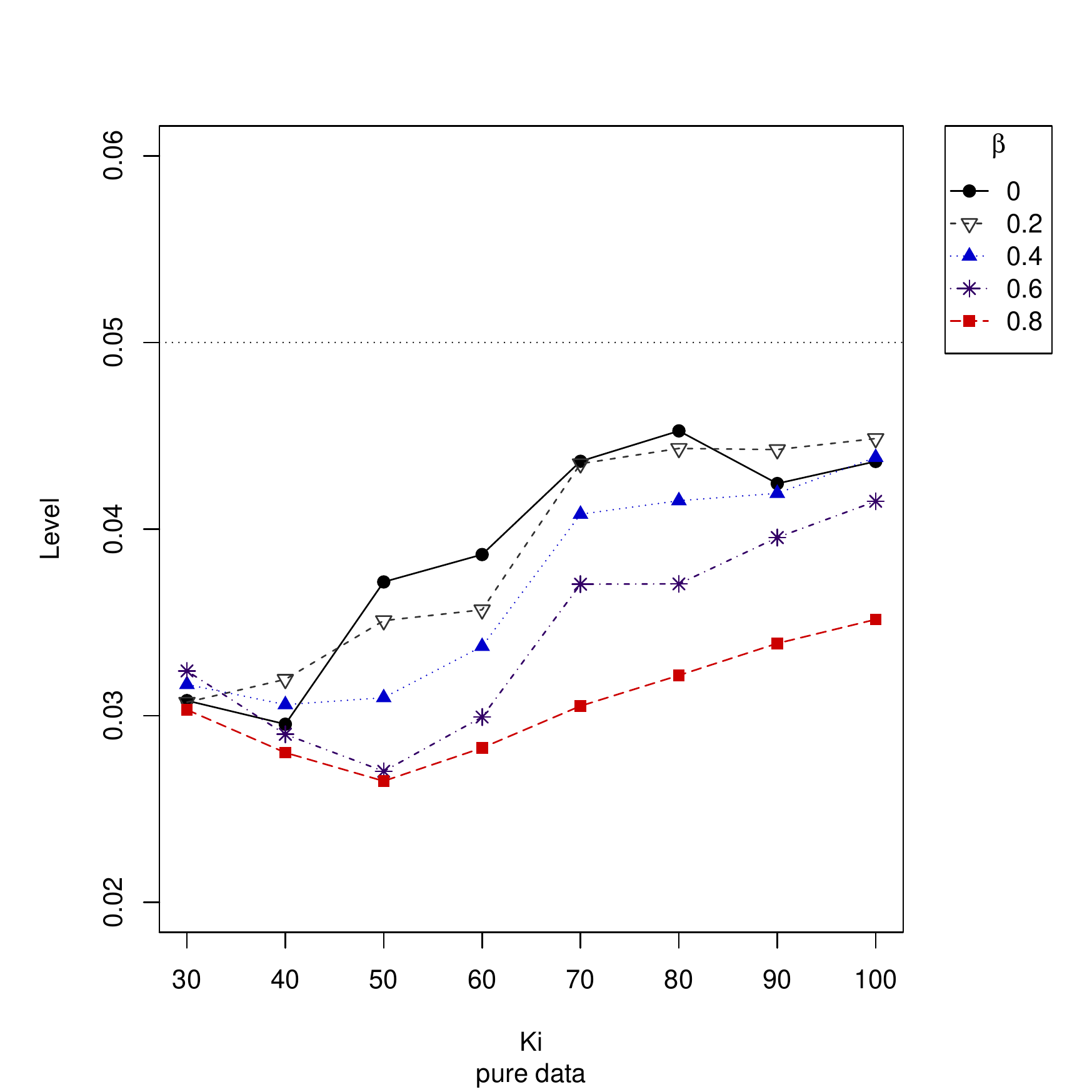} & 
\includegraphics[scale=0.4]{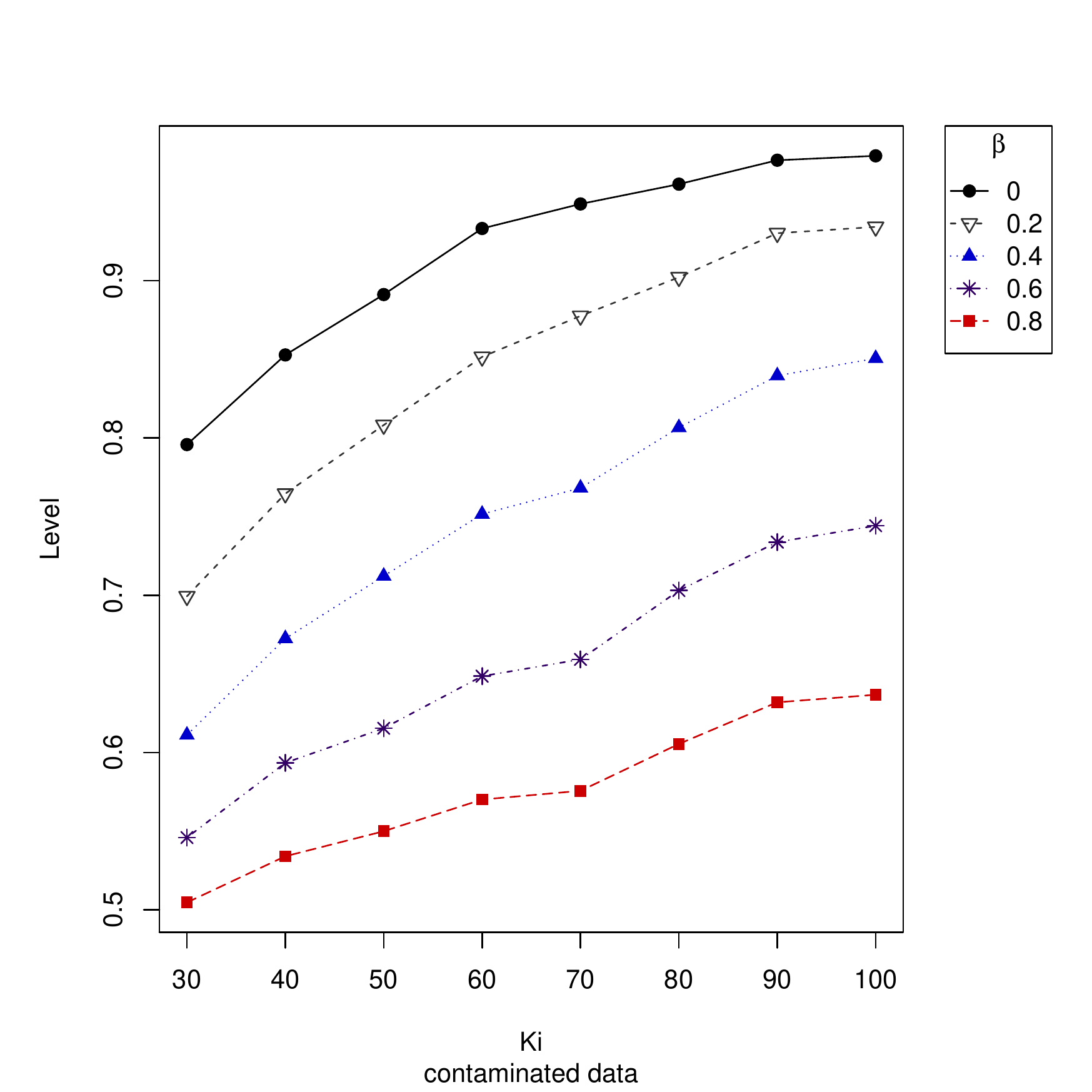} \\ 
\includegraphics[scale=0.4]{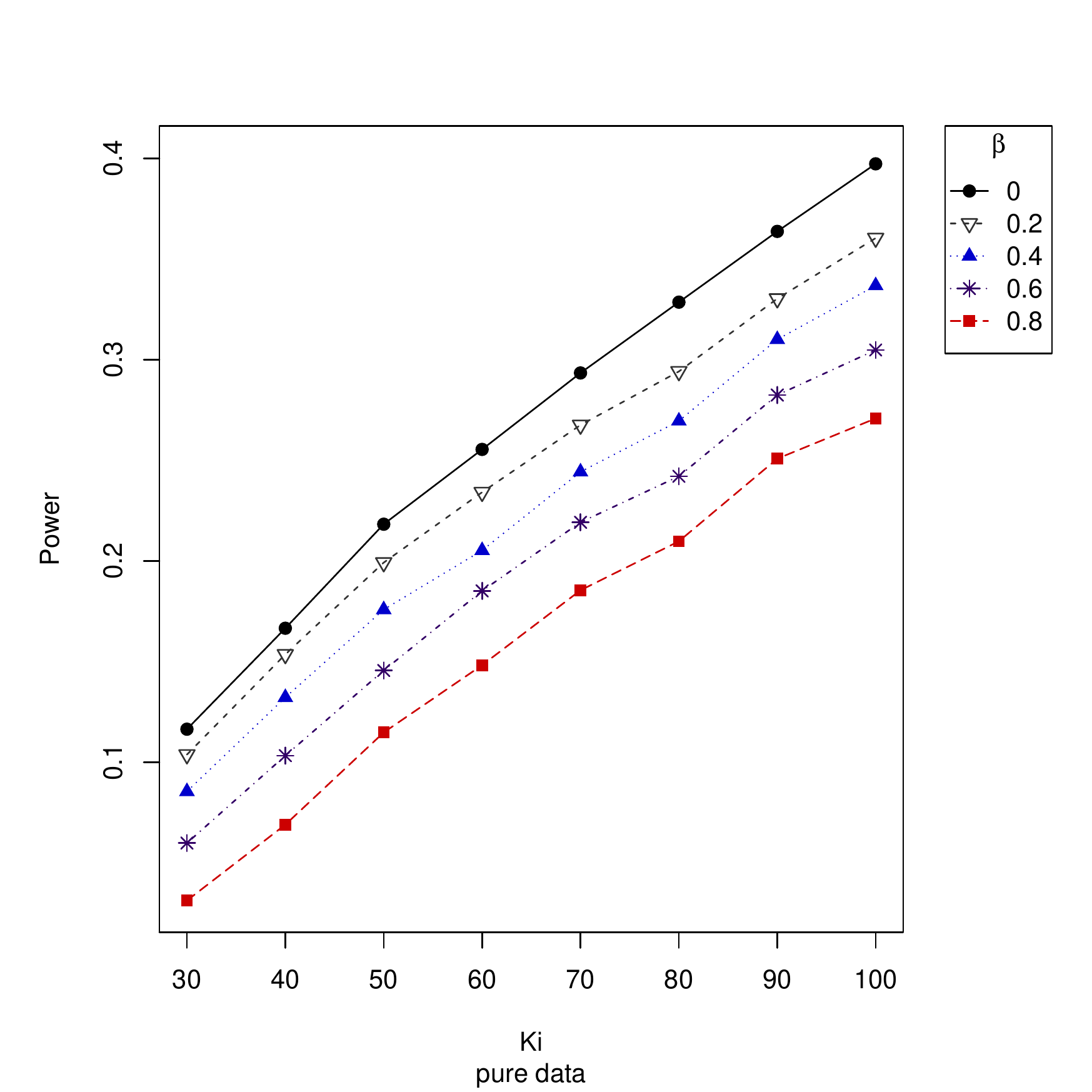} & 
\includegraphics[scale=0.4]{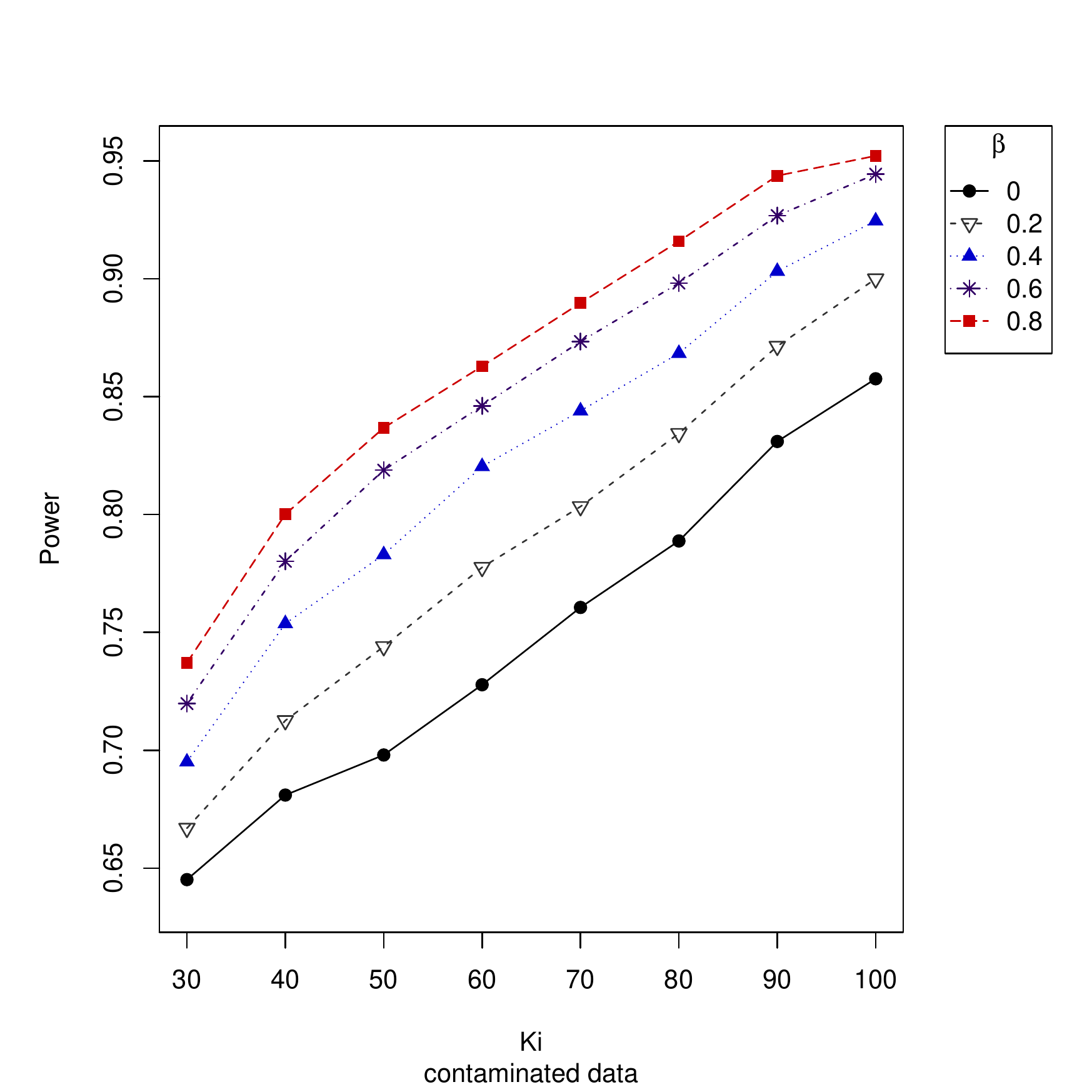} \\ 
\end{tabular}%
\caption{Levels and Powers of the weighted minimum density power divergence estimators-based Wald-type tests for different values of $K_i$ with pure (left) and contaminated data (right), for the case of moderate reliability.}
\label{fig:Wald}
\end{figure}

In the case of pure data,  we see how a big sample size is needed to obtain empirical tests close to the nominal level. In the case of contaminated data, empirical test levels are far away from the nominal level, with the MLE again presenting the least robust behaviour. 

This simulation study has illustrated well the robust properties of the weighted minimum density power divergence estimators for $\beta>0$, which is inevitably accompanied with a loss of efficiency in a the case of pure data.  It seems that a moderate low value of the tuning parameter can be a good choice when applying the estimators to a real data set. However, when dealing with specific data sets, especially when we have small data sets, a data driven procedure for the choice of tuning parameter will become necessary.

\subsection{Choice of  tuning parameter \label{sec:sim_choice}}
The problem of choosing the optimal tuning parameter in a DPD-based family of estimators has been extensively discussed in the literature; see, for example, Hong and Kim (2001), Warwick (2001), Warwick and Jones (2005), and Ghosh and Basu (2015). We now adopt the procedure proposed by Warwick and Jones (2005), which consists  minimizing the estimated mean square error of the estimators,  computed as the sum of estimated squared bias and variance; that is,
\begin{align*}
\widehat{MSE}_{\beta}&=(\widehat{\boldsymbol{\theta}}_{\beta}-\boldsymbol{\theta}_P)^T(\widehat{\boldsymbol{\theta}}_{\beta}-\boldsymbol{\theta}_P)  +\frac{1}{K}\text{trace}\left[  \boldsymbol{{J}}_{\beta}^{-1}(\widehat{\boldsymbol{\theta}}_{\beta})\boldsymbol{{K}}_{\beta}(\widehat{\boldsymbol{\theta}}_{\beta})\boldsymbol{{J}}_{\beta}^{-1}(\widehat{\boldsymbol{\theta}}_{\beta})  \right],
\end{align*}
where $\boldsymbol{\theta}_P$ is a pilot estimator, whose choice will be empirically discussed, since the overall procedure depends on this choice. If we take $\boldsymbol{\theta}_P=\widehat{\boldsymbol{\theta}}_{\beta}$, the approach  coincides with that of Hong and Kim (2001), but it does not take into account the model misspecification.

\begin{figure}[h!]
\center
\begin{tabular}{ll}
\includegraphics[scale=0.4]{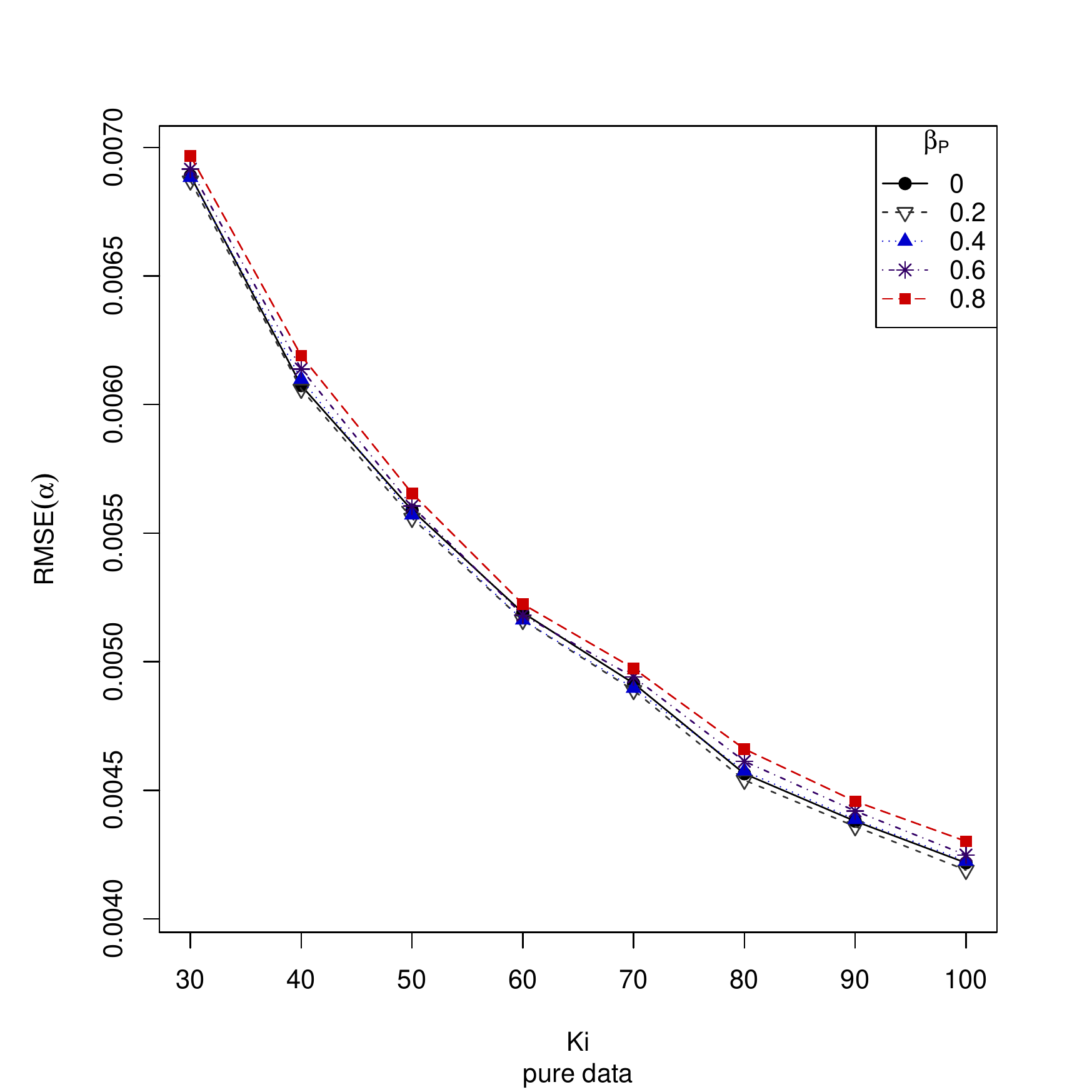} & 
\includegraphics[scale=0.4]{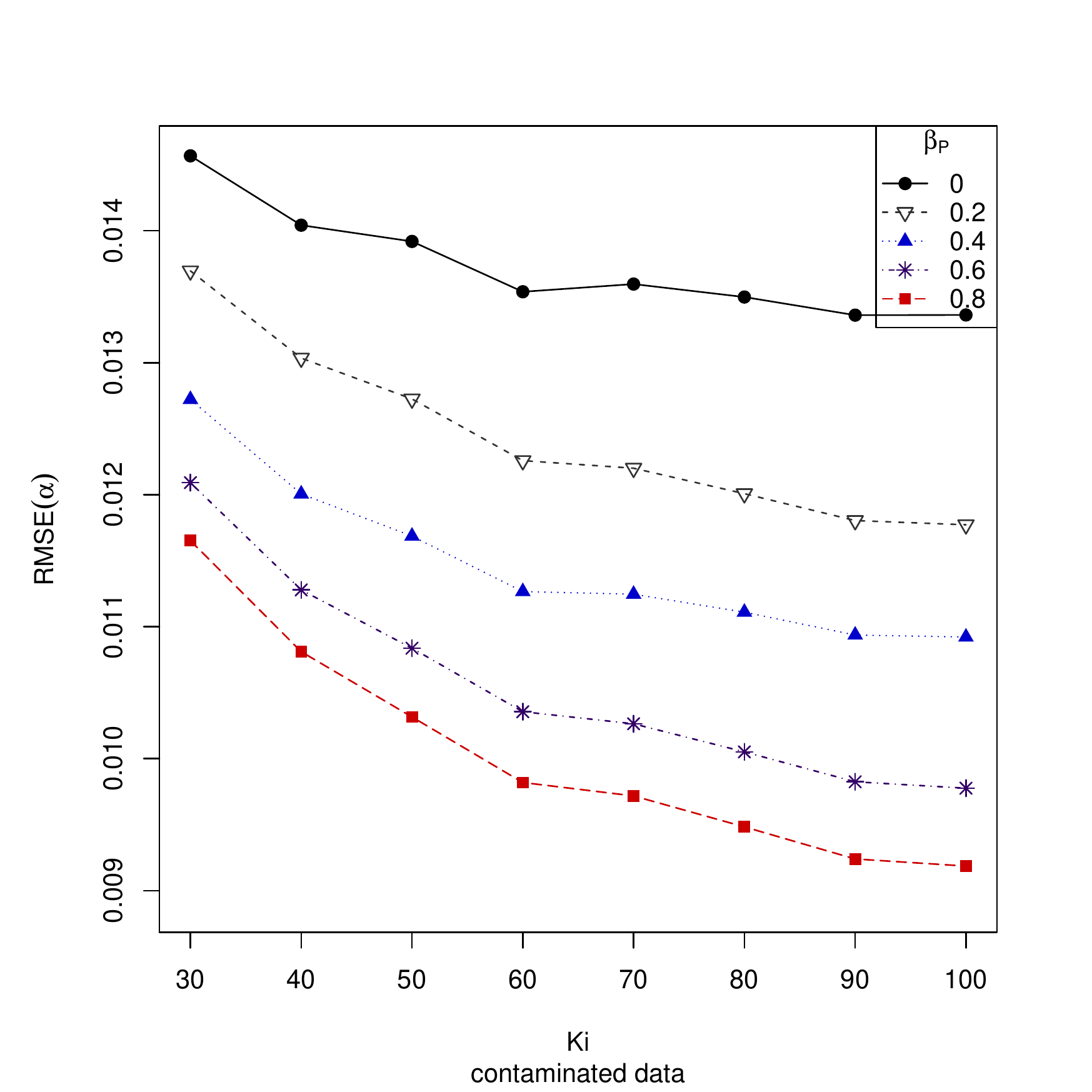} \\ 
\includegraphics[scale=0.4]{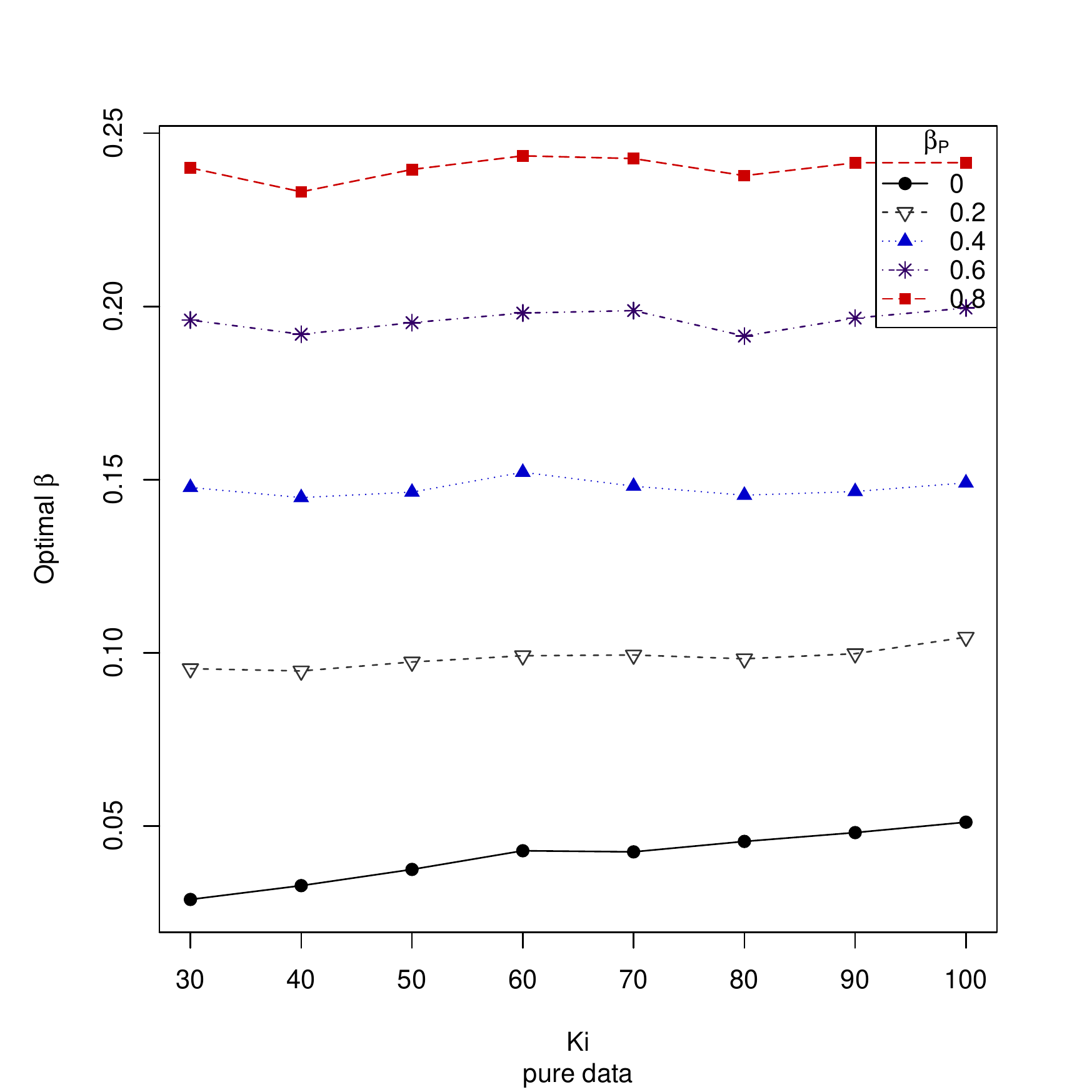} & 
\includegraphics[scale=0.4]{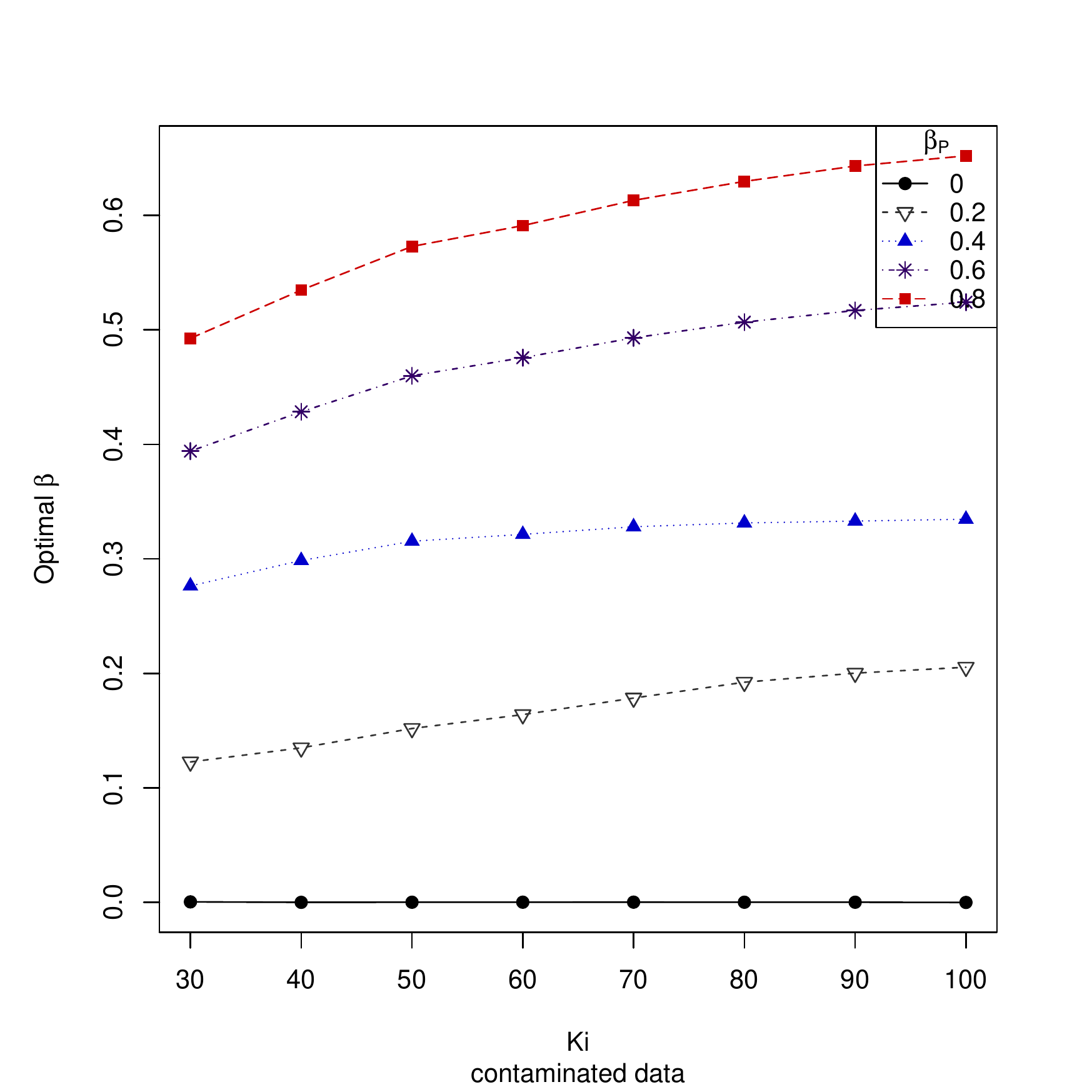} \\ 
\end{tabular}%
\caption{Estimated  optimal $\beta$ and the corresponding RMSEs for different pilot estimators in the proposed ad-hoc approach for the case of moderate reliabiulity.}
\label{fig:choice}
\end{figure}

We consider again the balanced scenario under moderate reliability discussed earlier.  For different pilot estimators and a grid of $100$ points, optimal tuning parameters and their corresponding RMSEs are computed.  The optimal tuning parameter increases when the contamination level increases in the data, and it seems that a moderate value of $\beta$ is the best choice for the pilot estimator, as suggested in the work of Warwick and Jones (2005). 

\begin{figure}[h!]
\centering
\begin{tabular}{cc}
\includegraphics[scale=0.41]{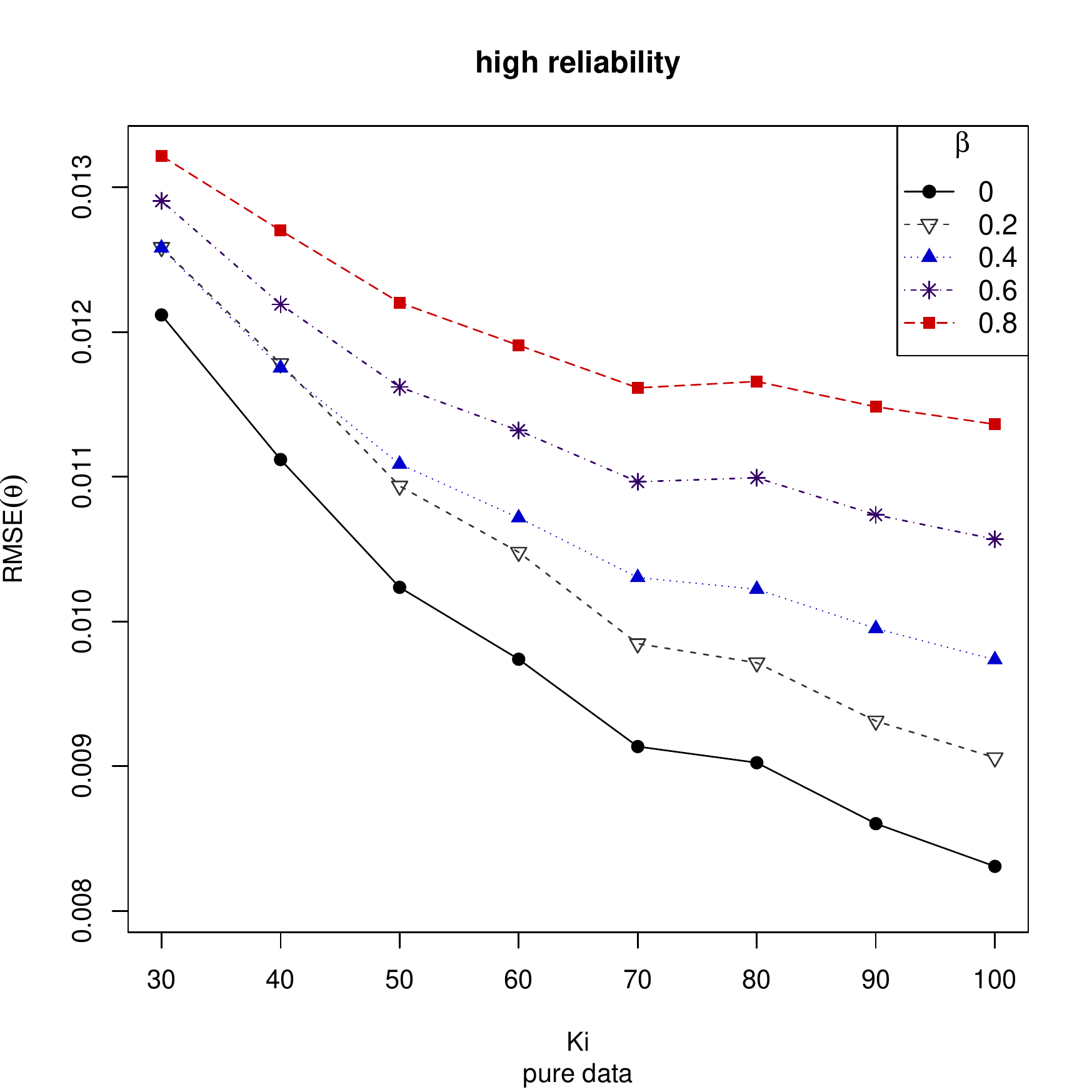} & 
\includegraphics[scale=0.41]{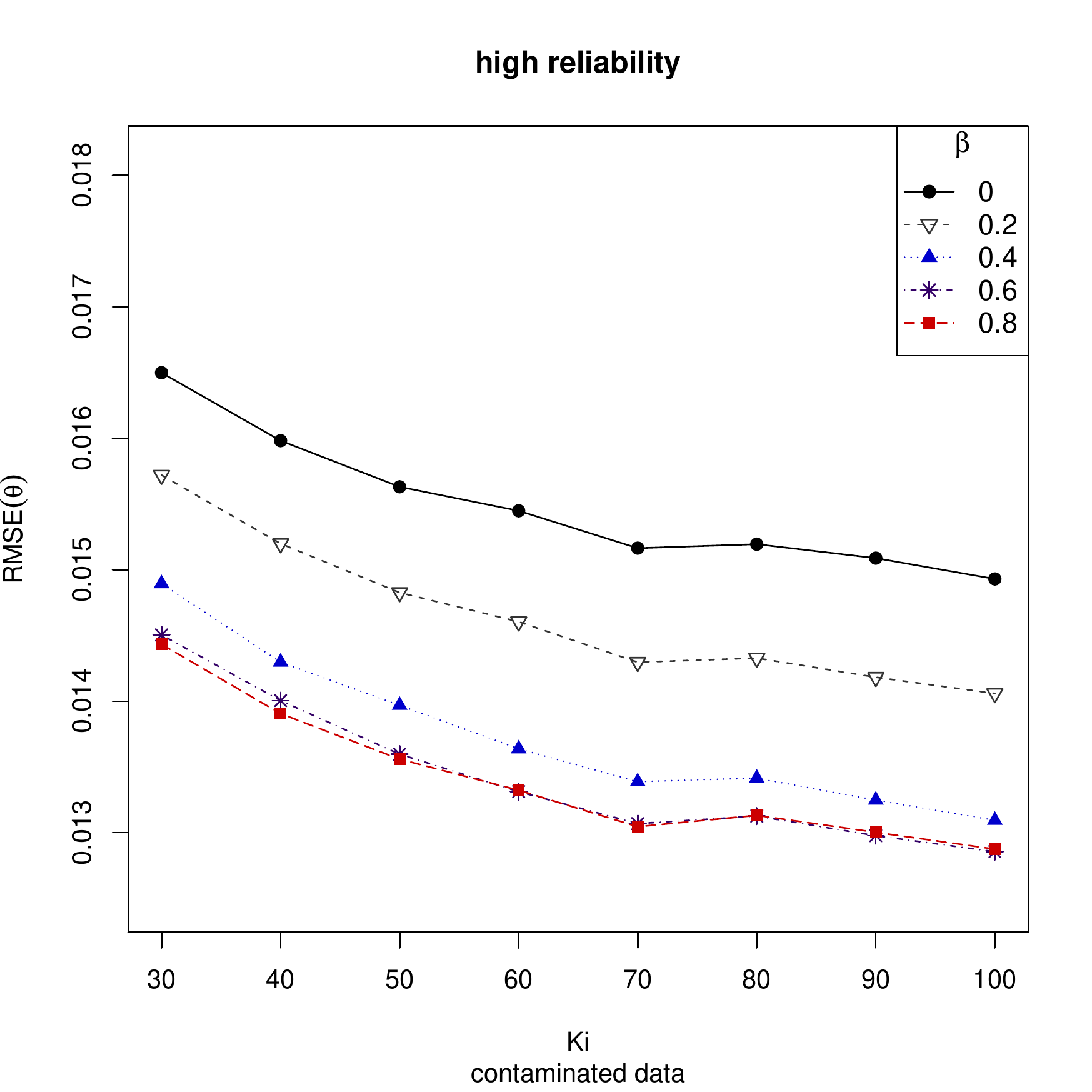} \\ 
\includegraphics[scale=0.41]{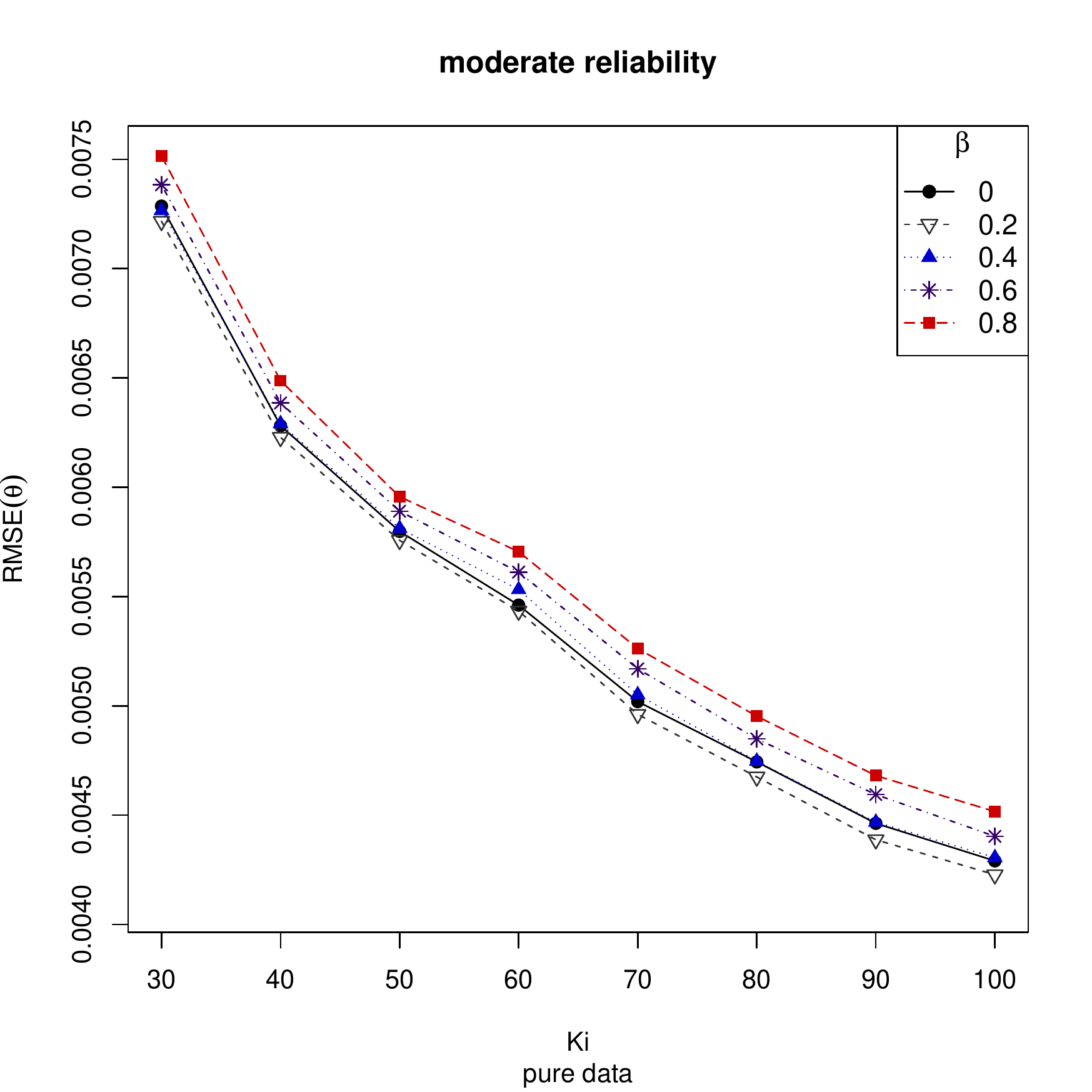} & 
\includegraphics[scale=0.41]{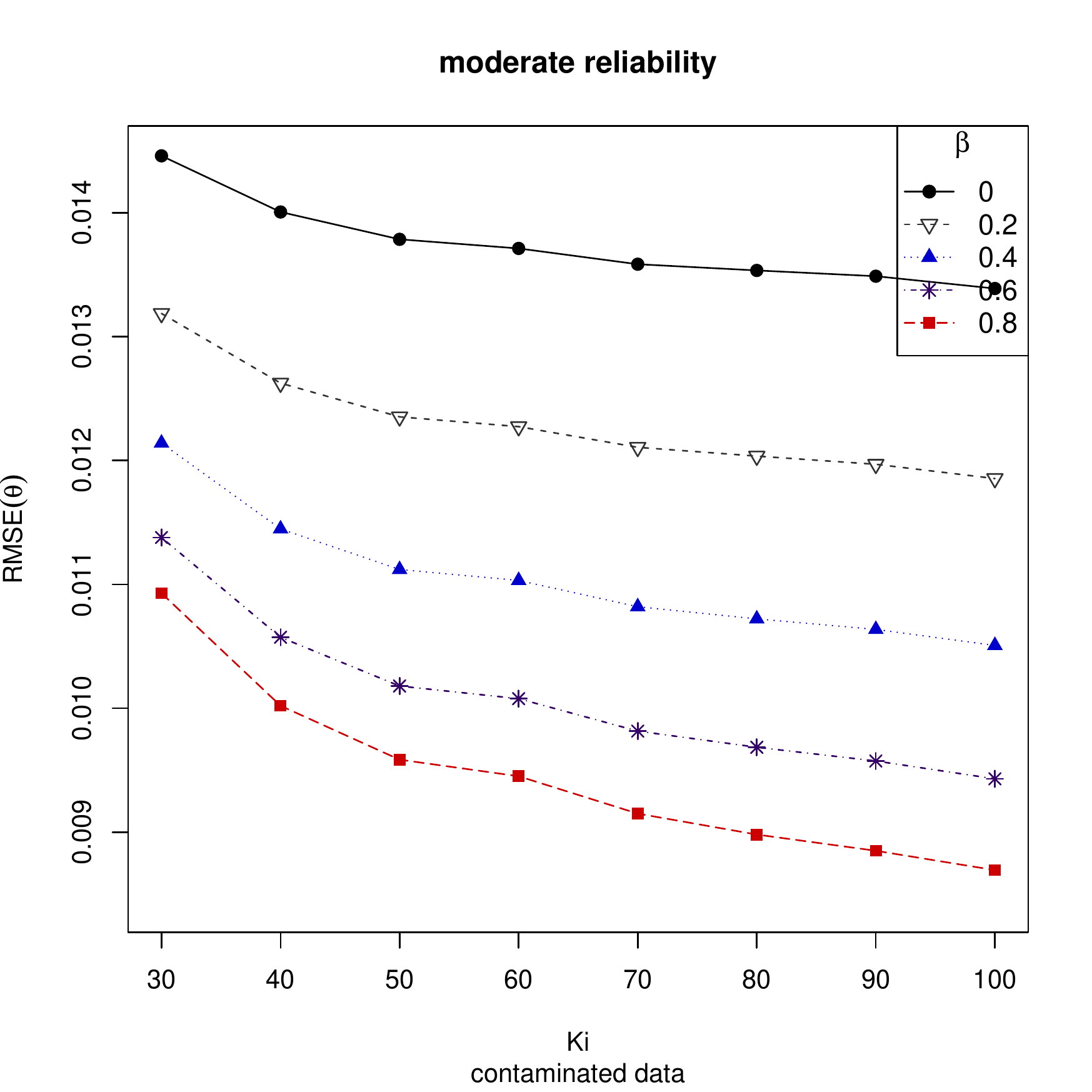} \\ 
\includegraphics[scale=0.41]{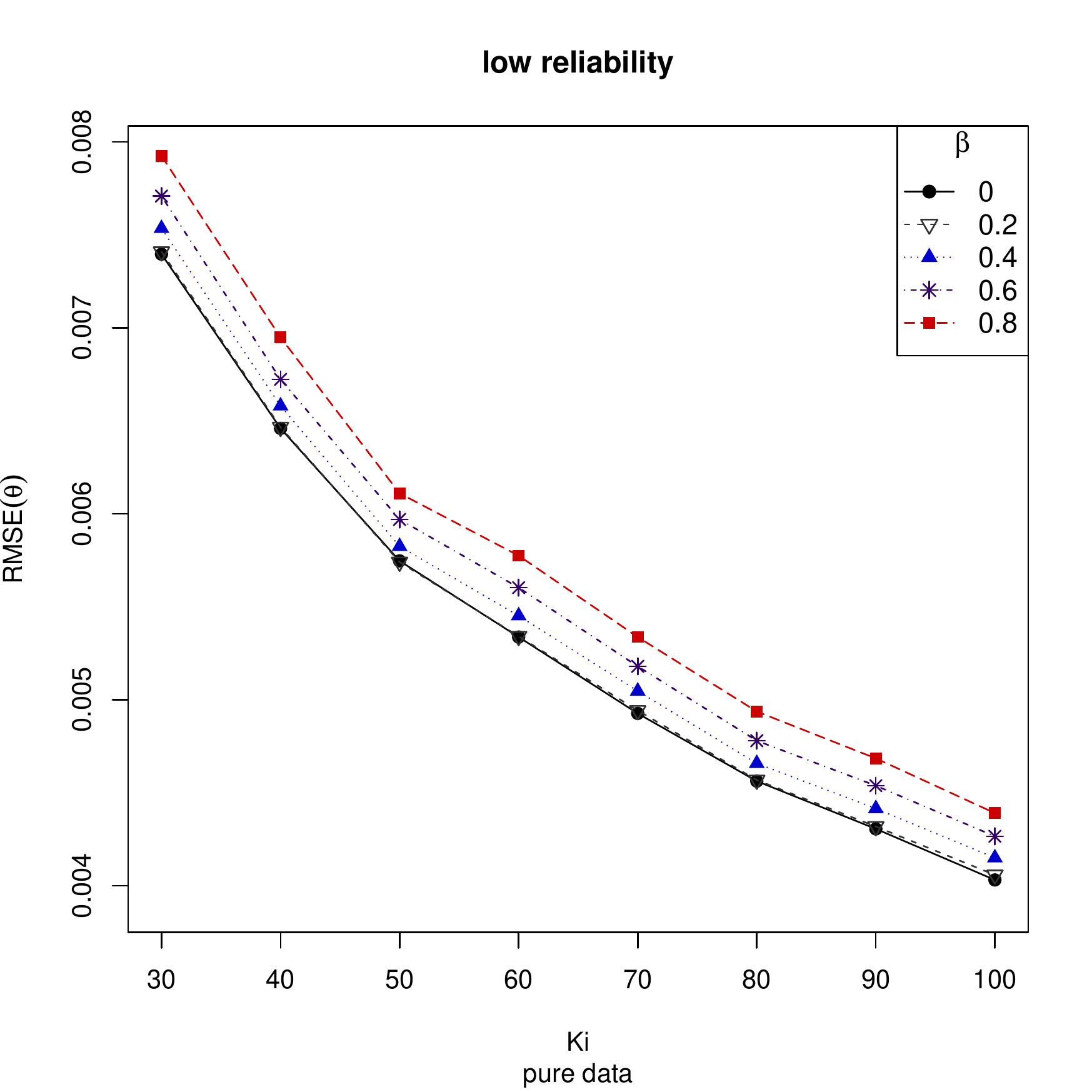} & 
\includegraphics[scale=0.41]{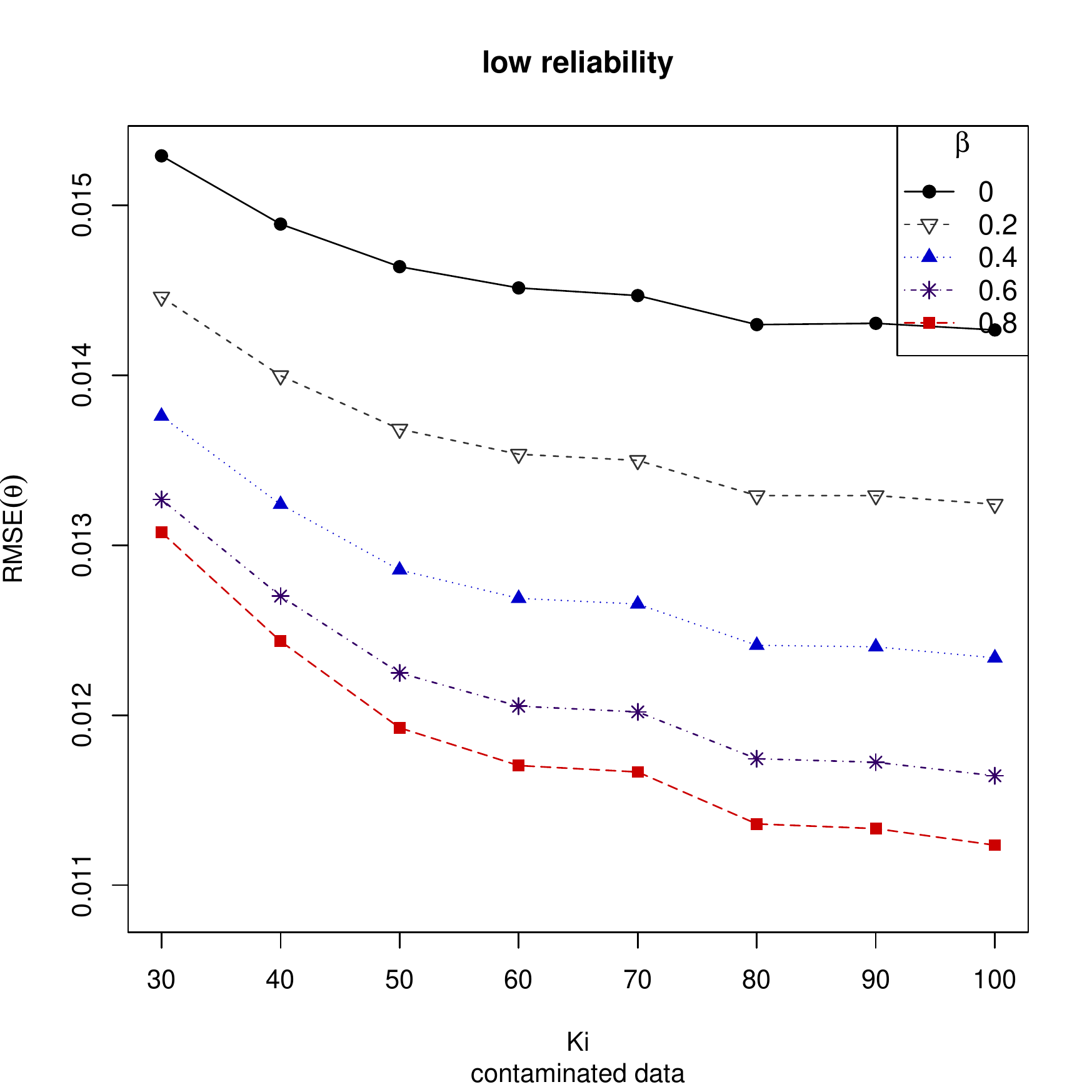}%
\end{tabular}%
\caption{RMSEs  of the weighted minimum density power divergence estimators of $\boldsymbol{\theta }$ for different values of reliability with pure (left) and contaminated data (right)}
\label{fig:RMSE}
\end{figure}

\begin{figure}[h!]
\centering
\begin{tabular}{cc}
\includegraphics[scale=0.41]{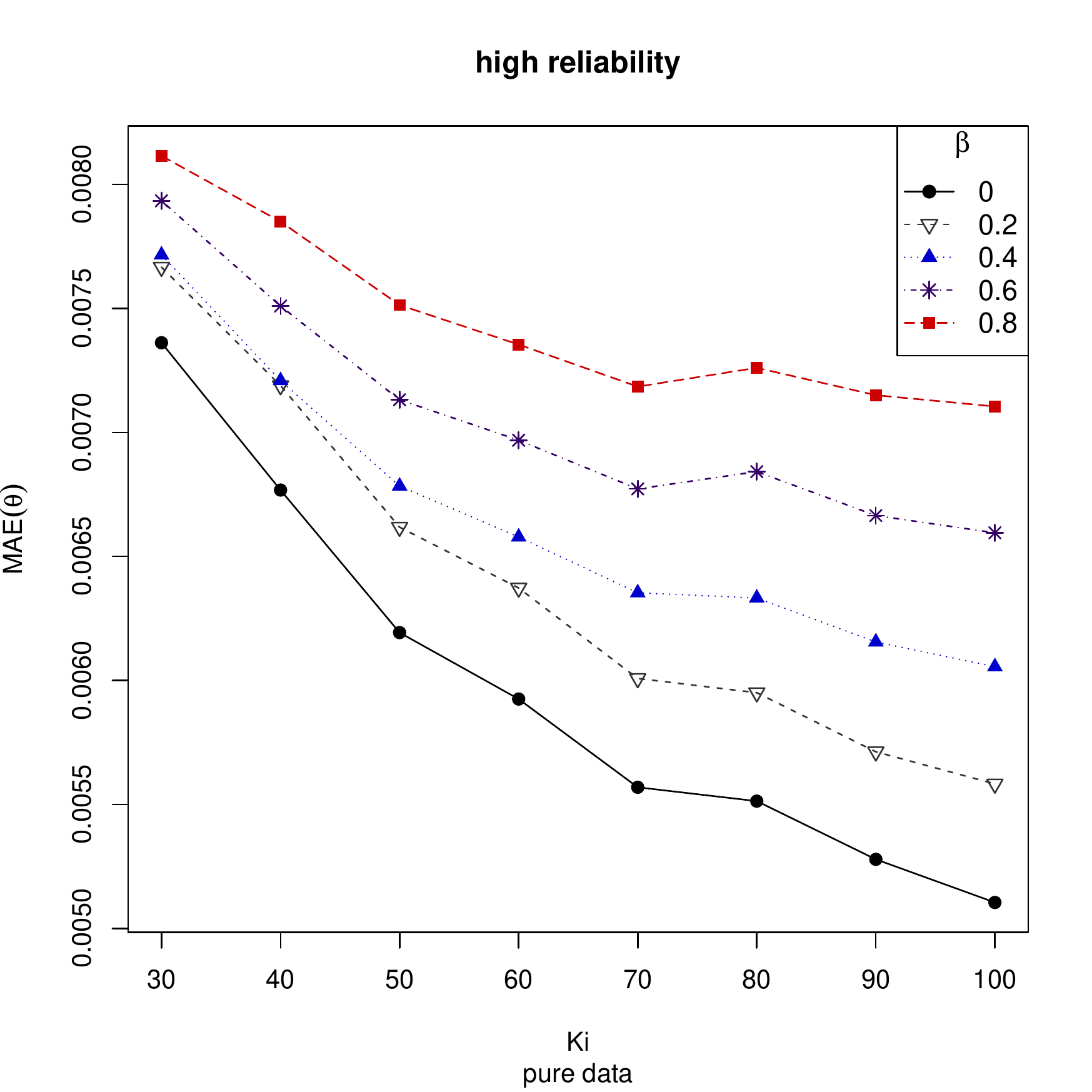} & 
\includegraphics[scale=0.41]{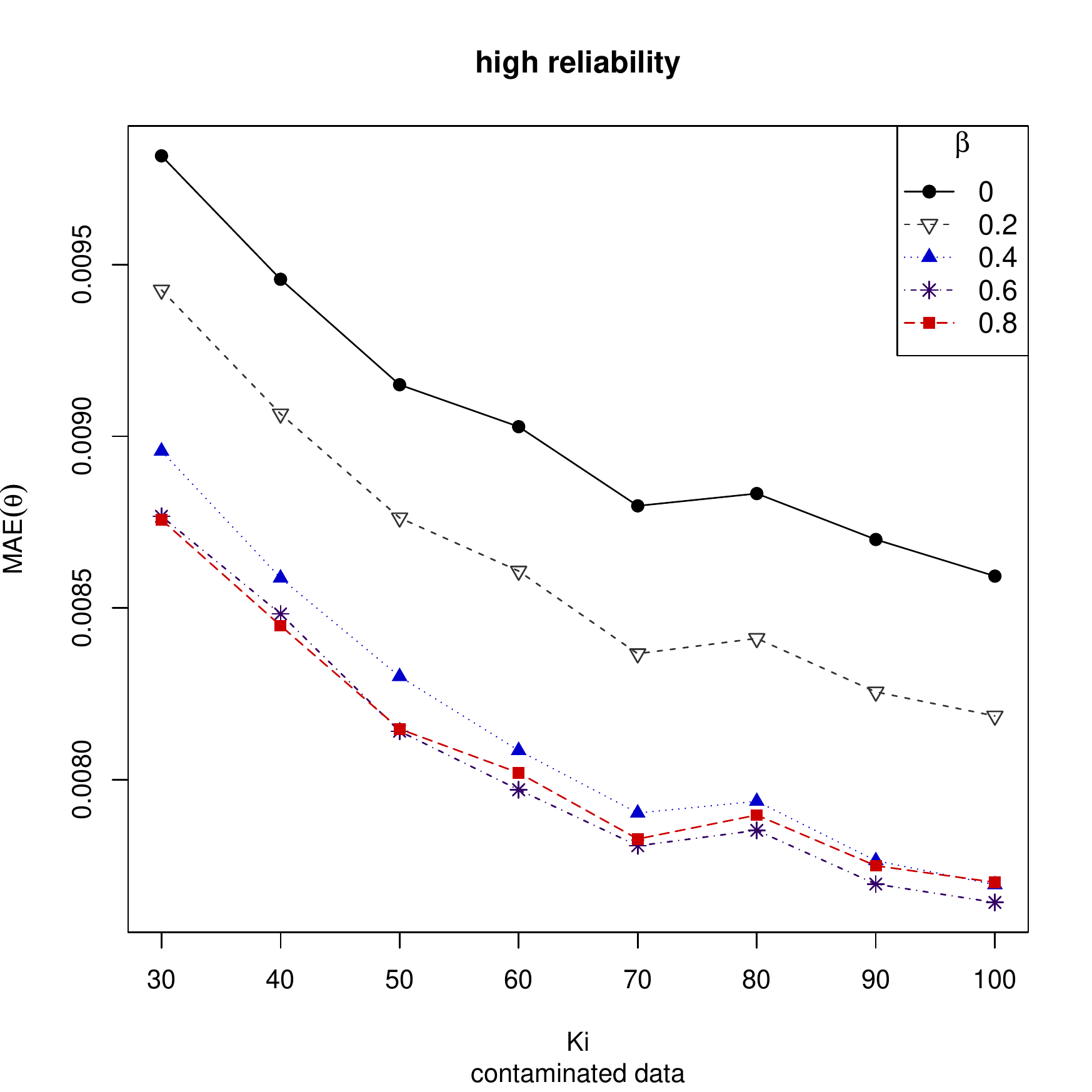} \\ 
\includegraphics[scale=0.41]{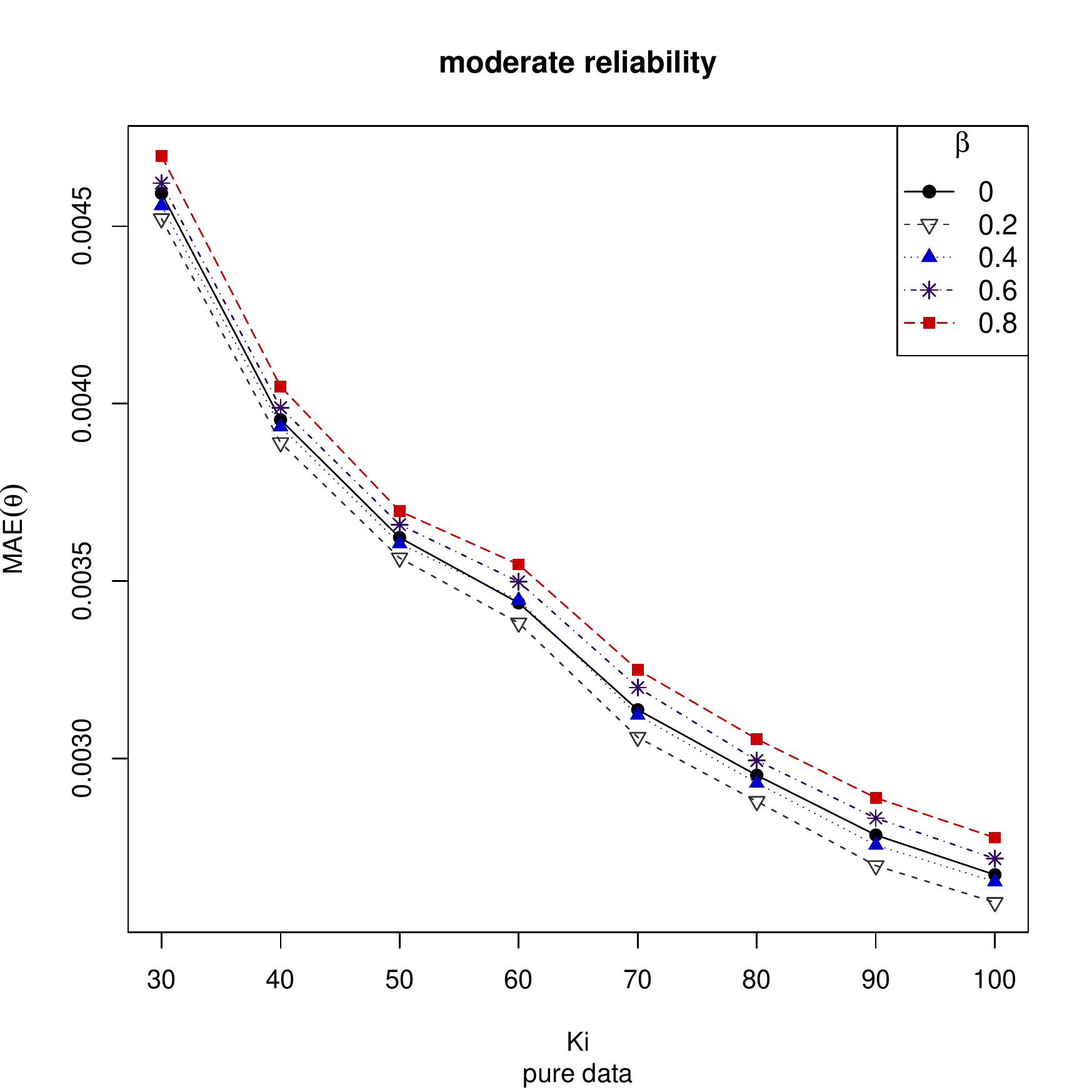} & 
\includegraphics[scale=0.41]{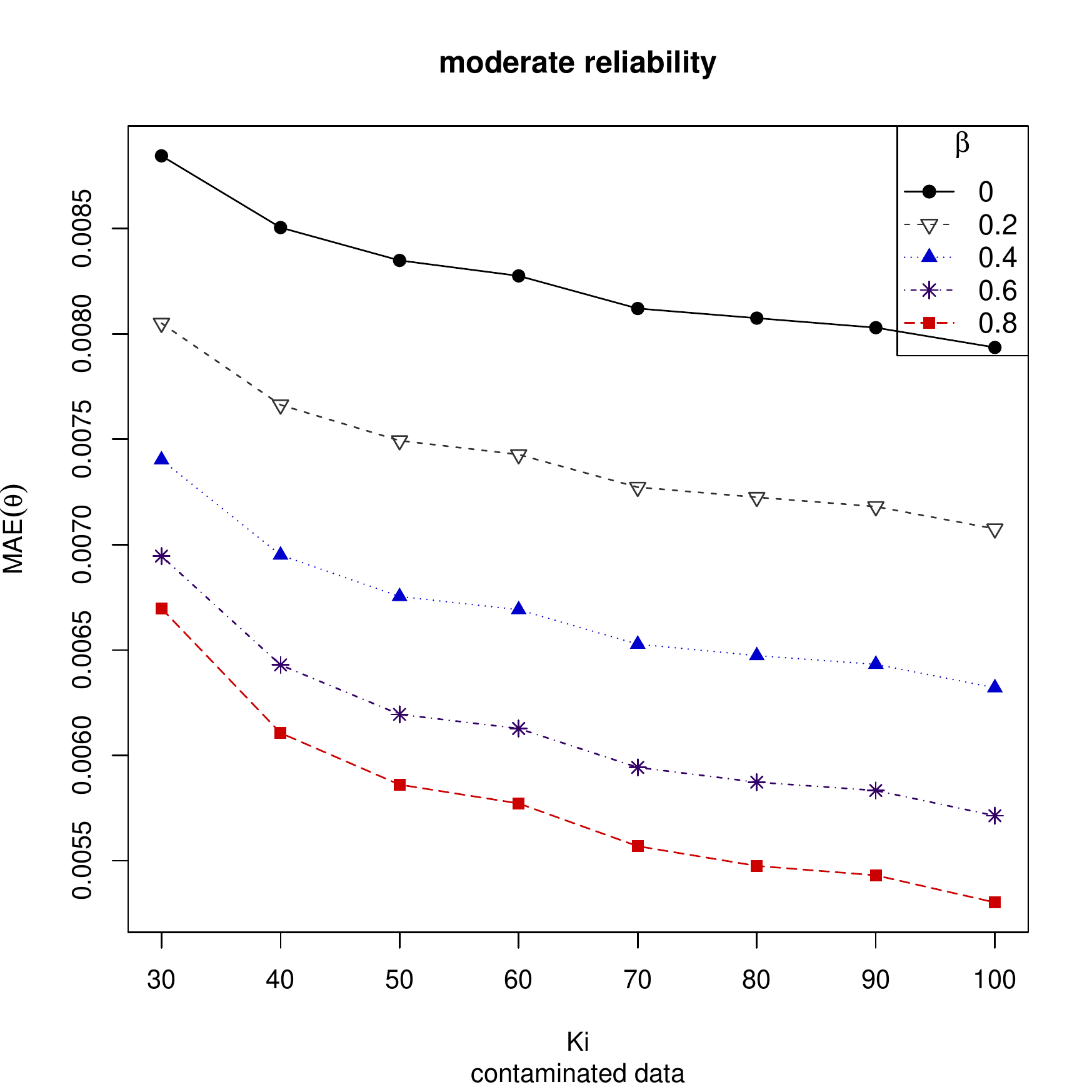} \\ 
\includegraphics[scale=0.41]{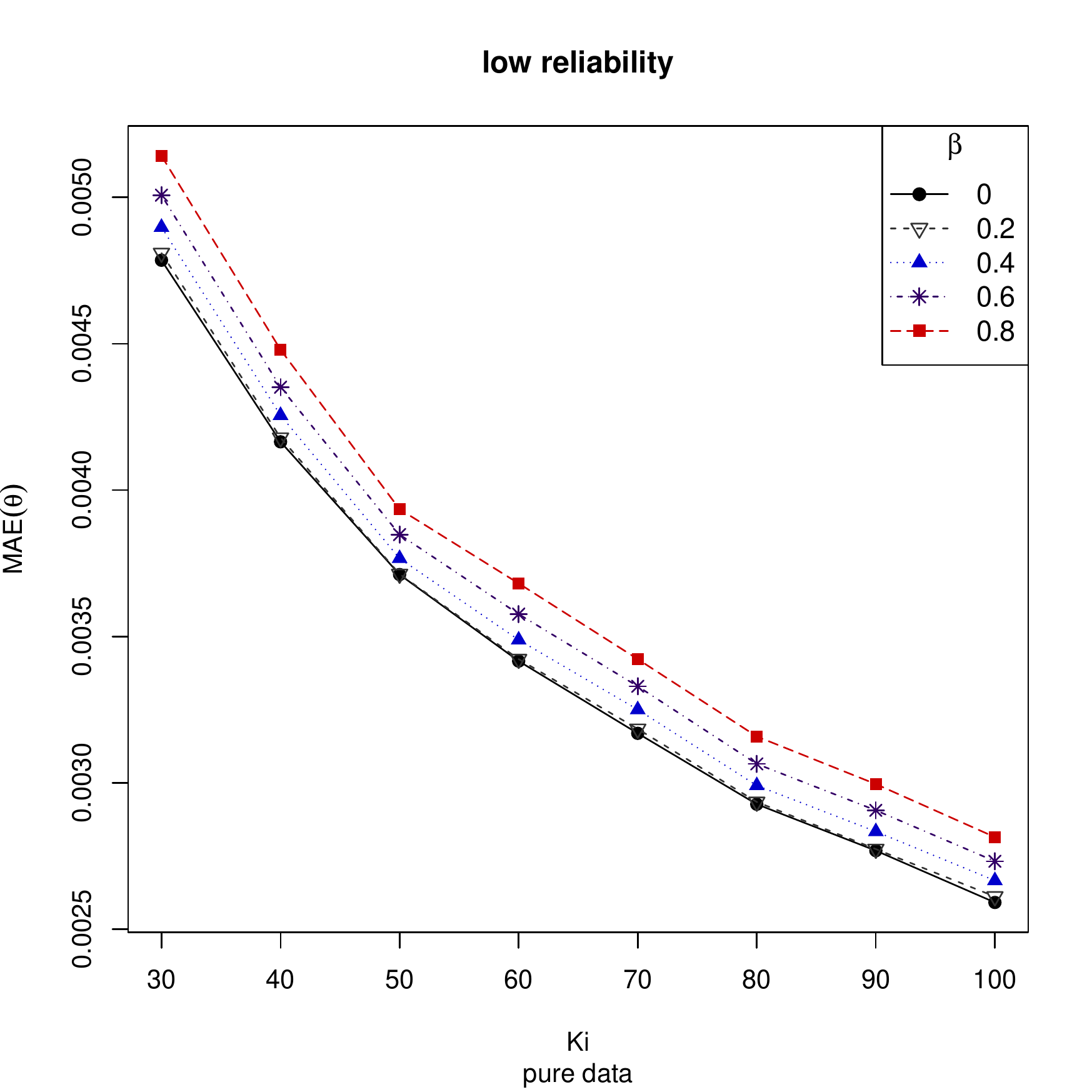} & 
\includegraphics[scale=0.41]{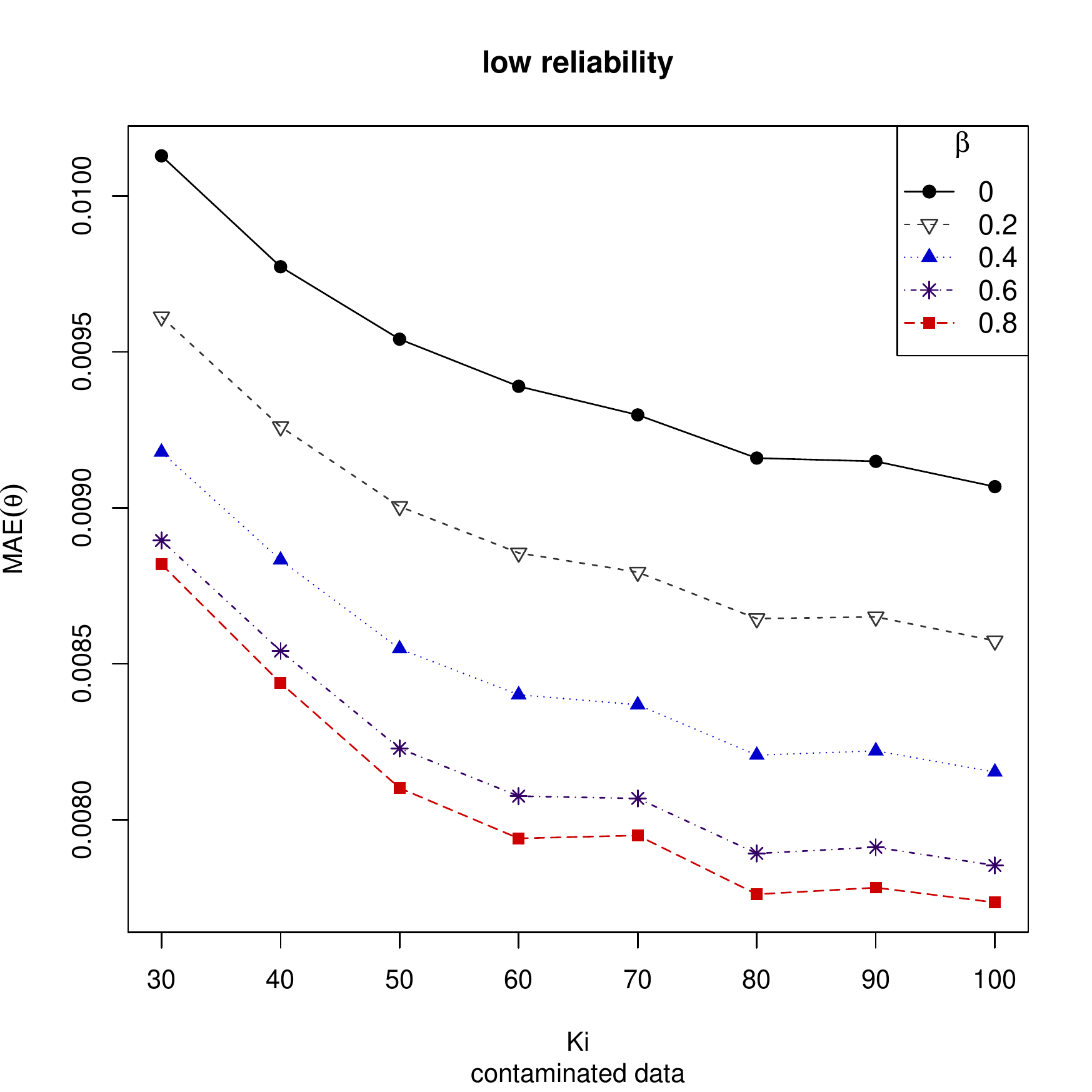}%
\end{tabular}%
\caption{MAEs  of the weighted minimum density power divergence estimators of $\boldsymbol{\theta }$ for different values of reliability with pure (left) and contaminated data (right).}
\label{fig:MAE}
\end{figure}

\begin{figure}[h!]
\centering
\begin{tabular}{cc}
\includegraphics[scale=0.41]{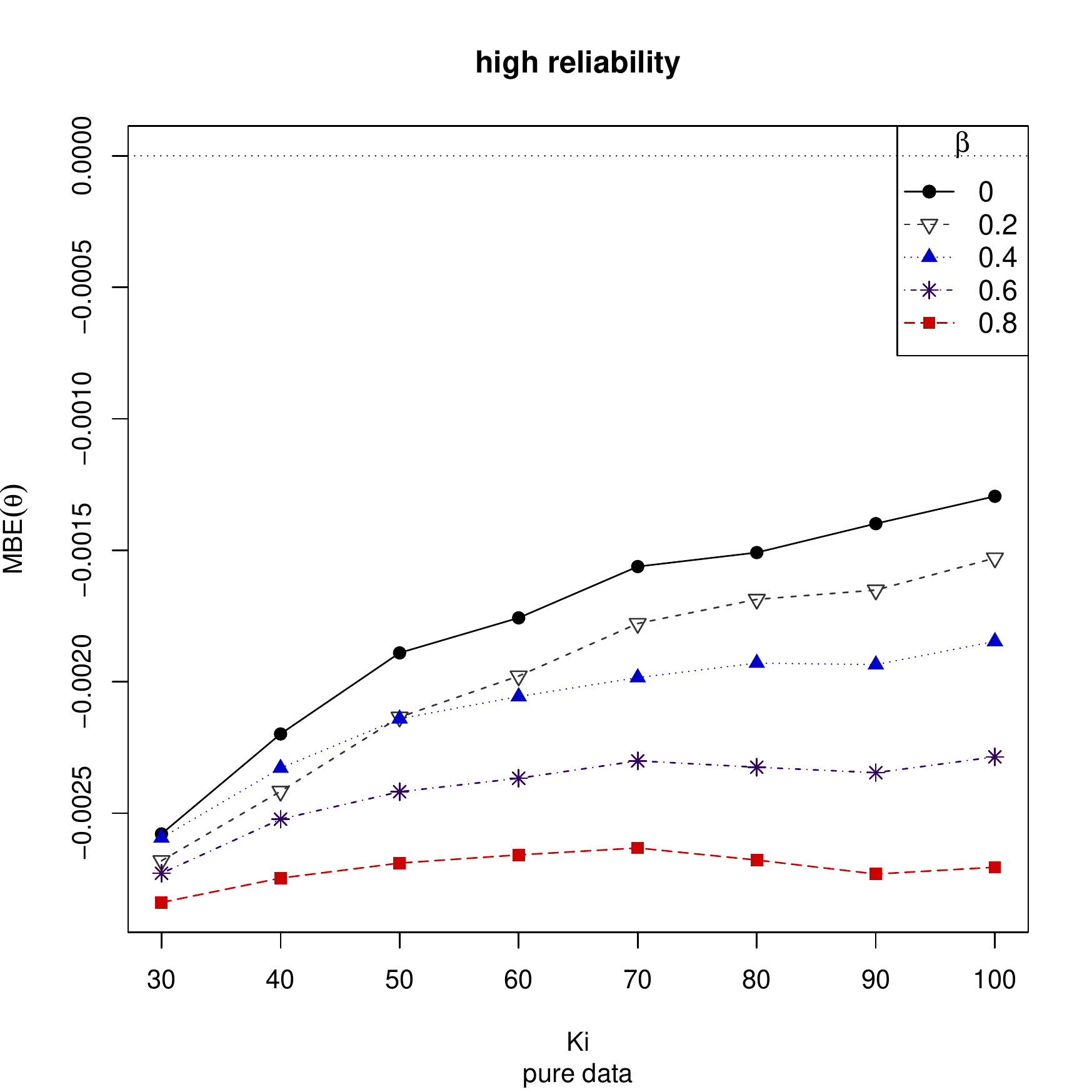} & 
\includegraphics[scale=0.41]{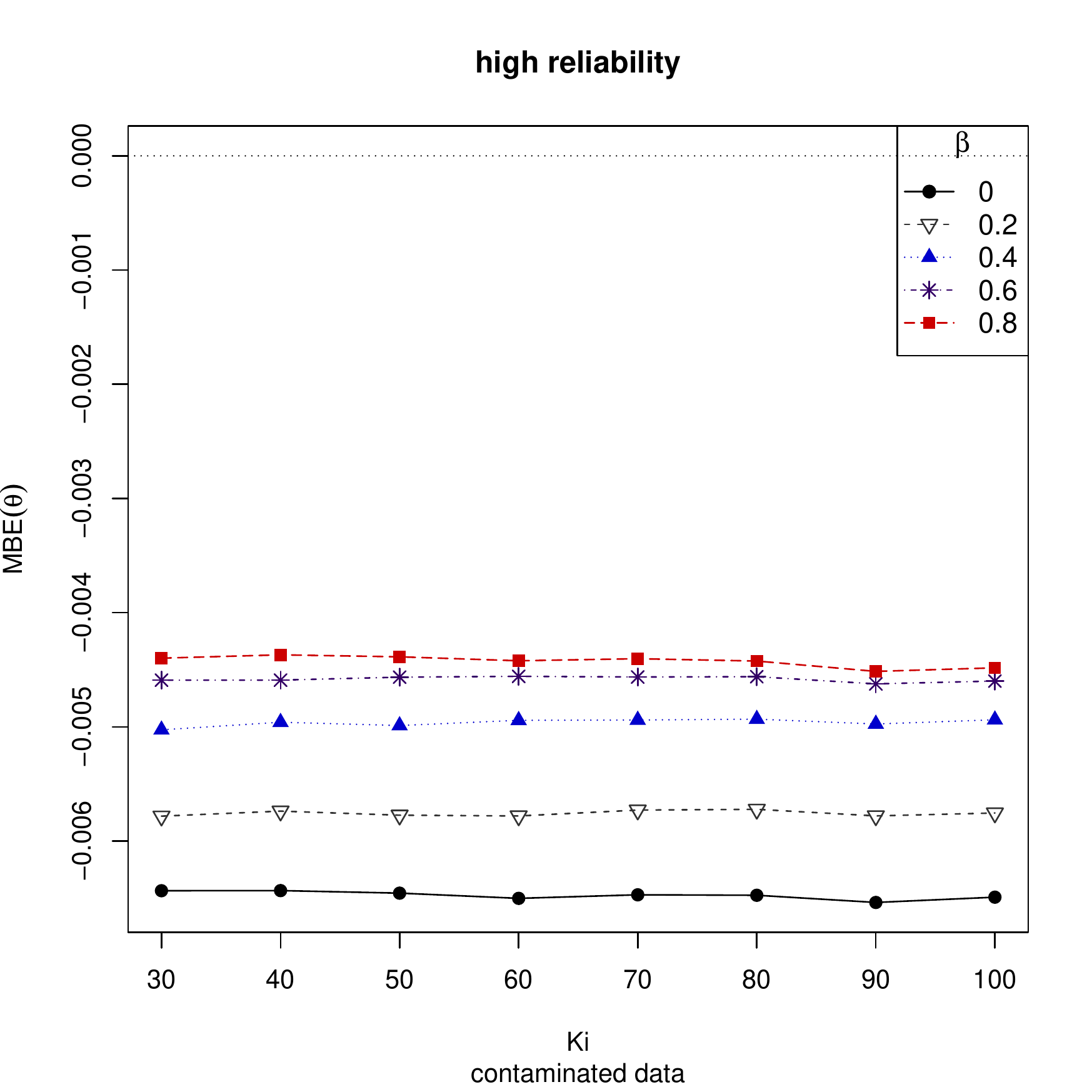} \\ 
\includegraphics[scale=0.41]{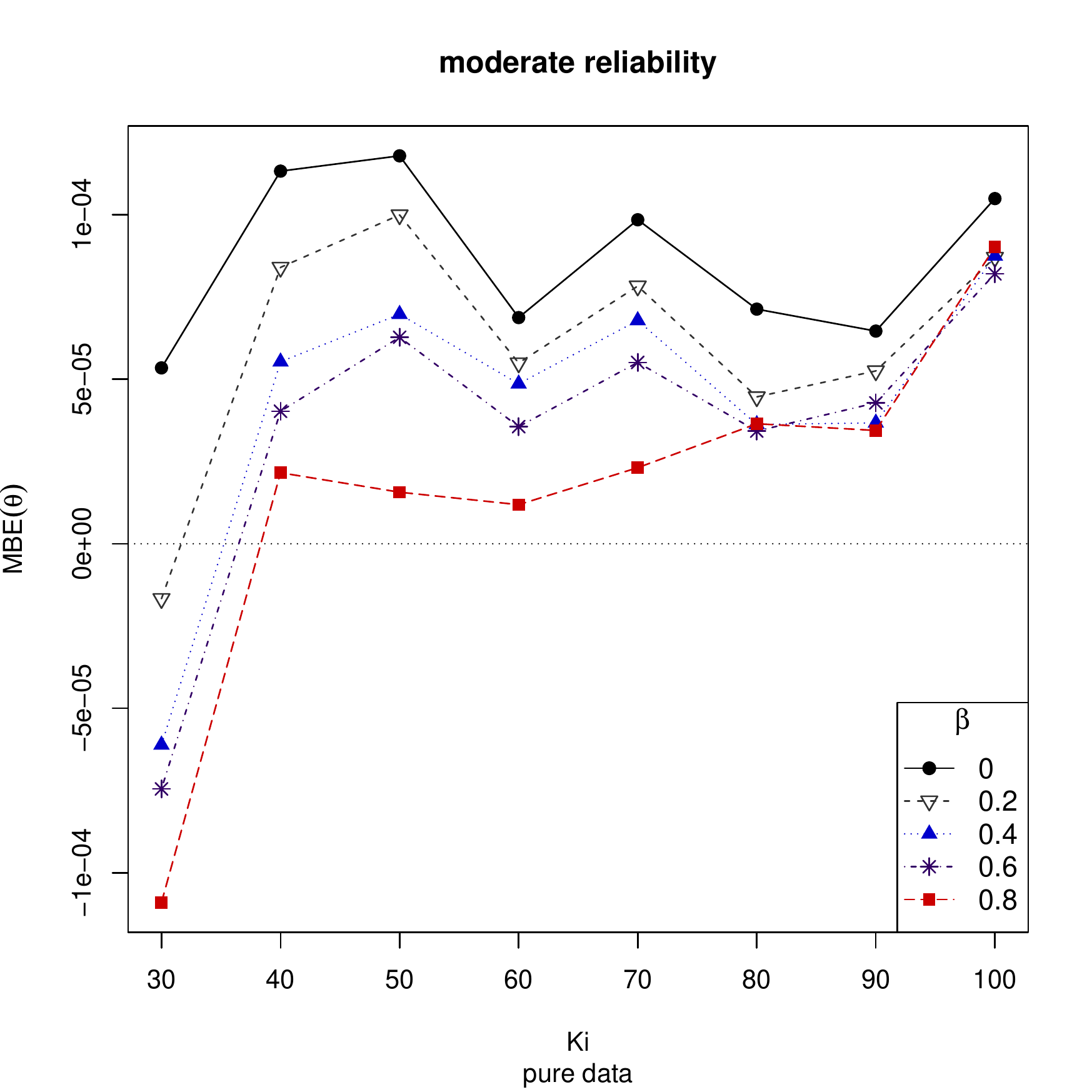} & 
\includegraphics[scale=0.41]{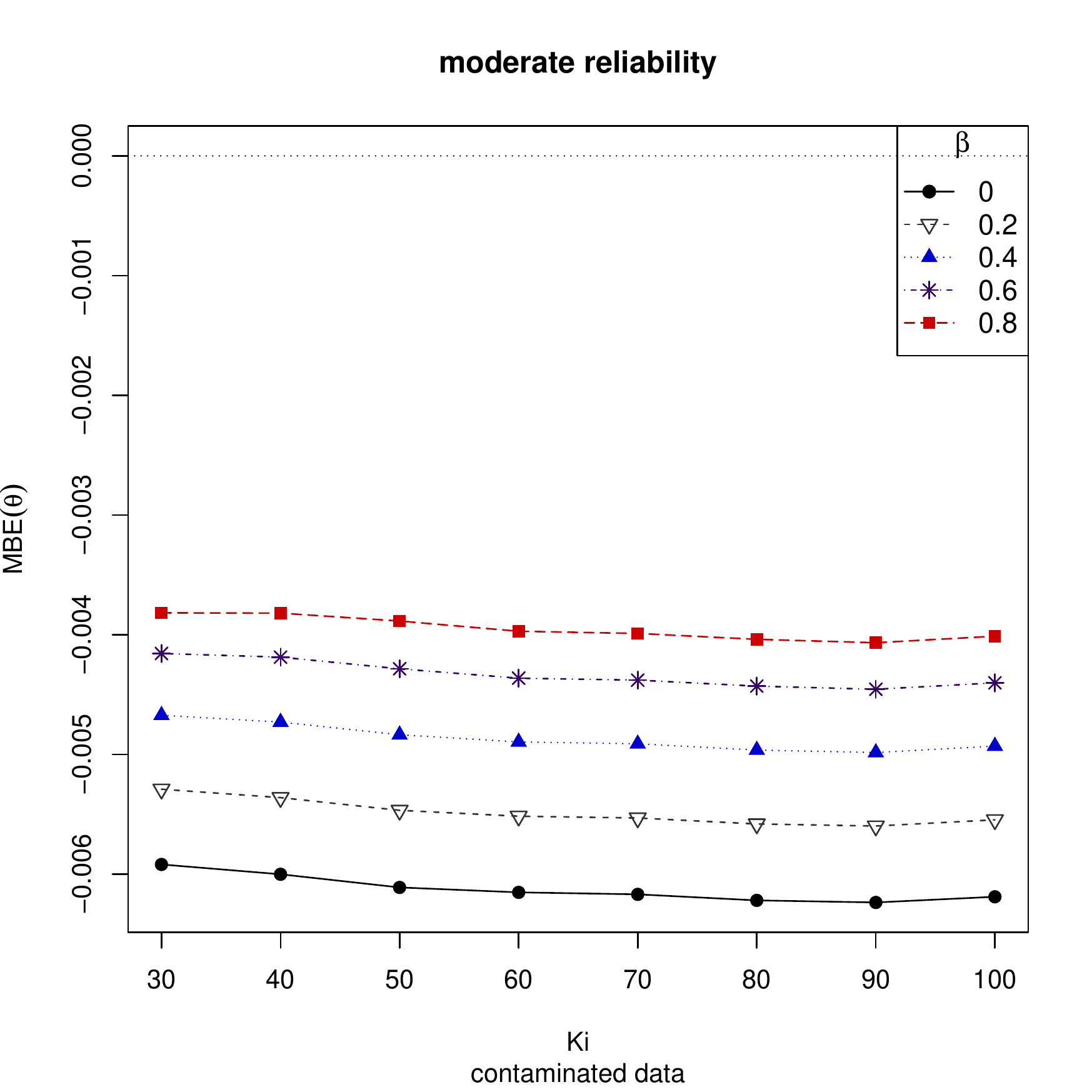} \\ 
\includegraphics[scale=0.41]{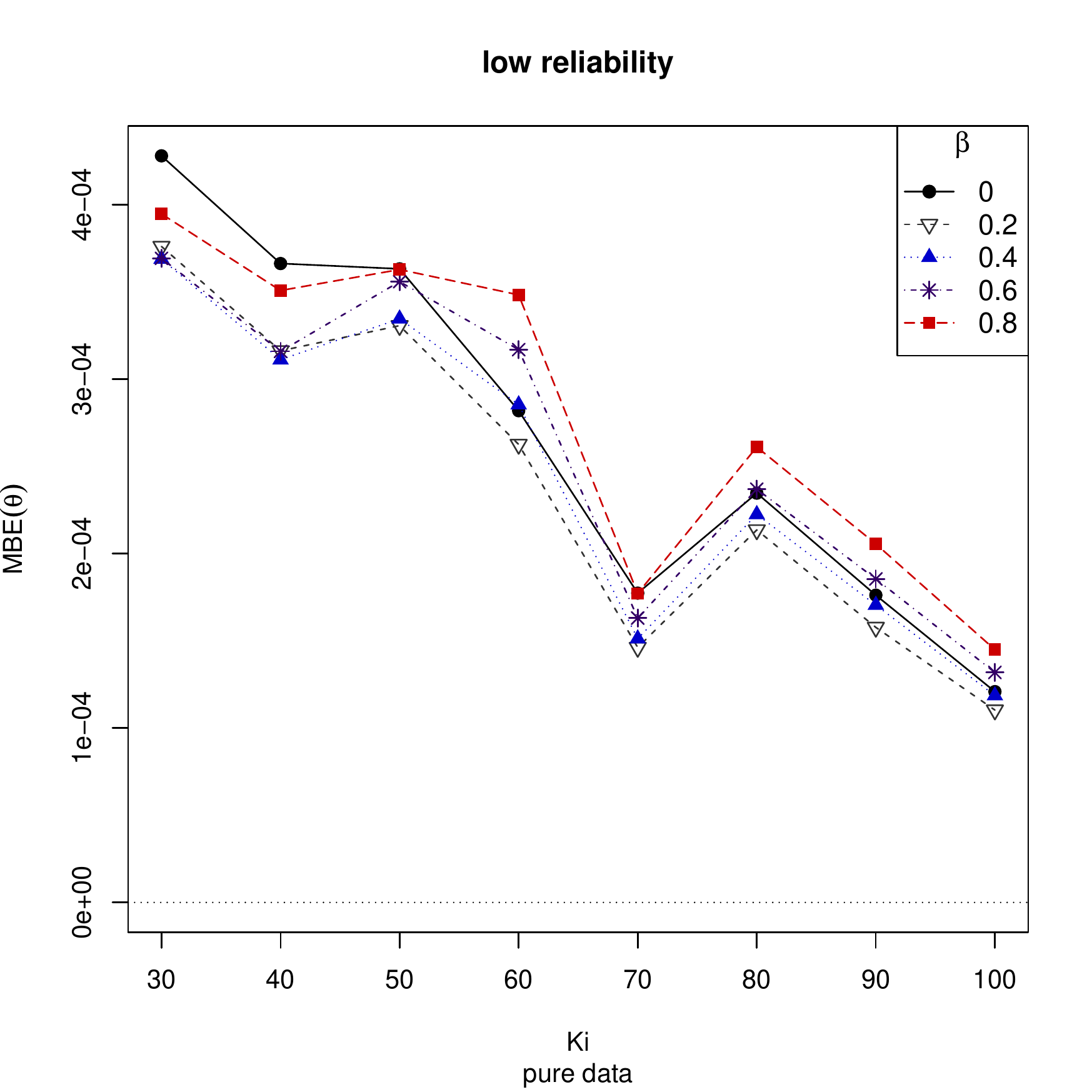} & 
\includegraphics[scale=0.41]{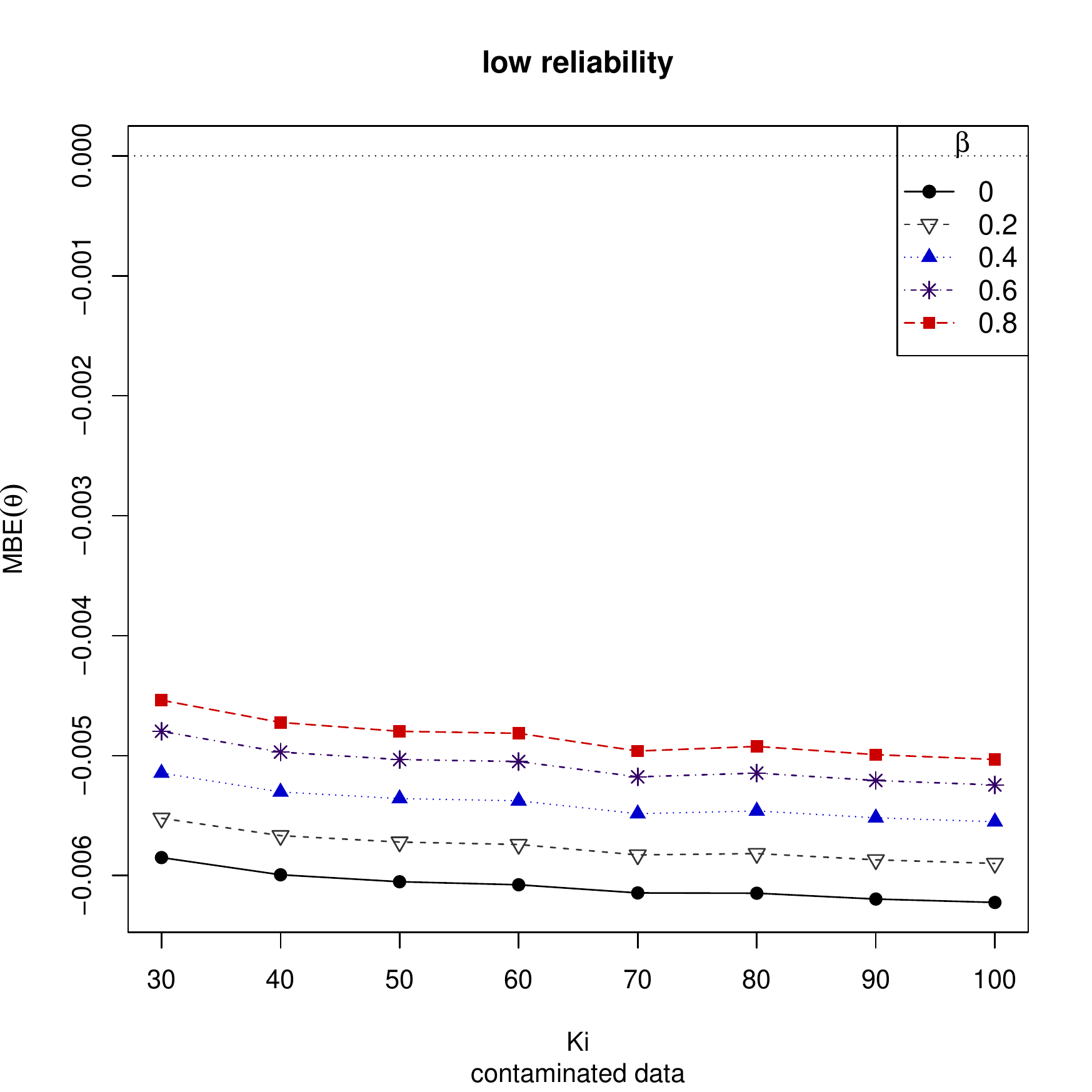}%
\end{tabular}%
\caption{MBEs  of the weighted minimum density power divergence estimators of $\boldsymbol{\theta }$ for different values of reliability with pure (left) and contaminated data (right).}
\label{fig:MBE}
\end{figure}

\clearpage
\section{Benzidine dihydrochloride (BDC) experiment \label{sec:num}}
Let us reconsider our motivating example, the BDC experiment, to study the performance of the proposed  procedures.  As noted in Section \ref{sec:model}, this experiment considers two different  doses of drug  induced in the mice: 60 parts per million ($\omega=1$) and 400 parts per million ($\omega=2$) and two causes of death are recorded: died without tumor ($\delta_{ijk}=1$) and died with tumor ($\delta_{ijk}=2$). The data are presented in Table \ref{table:numdata}.

\begin{table}[h!]
\caption{Estimations for the BDC experiment for different hoices of tuning parameters\label{table:numdata2}} \setlength{\tabcolsep}{5pt}
\small
\begin{tabular}{l rrrr rr rrr r}
\hline
$\beta$&$\theta_{10}$ & $\theta_{11}$ & $\theta_{20}$ & $\theta_{21}$ & $E_{w=1}^1$  & $E_{w=2}^1$    & $E_{w=1}$     & $E_{w=2}$     & $P_{w=1}^1$   & $P_{w=2}^1$  \\  \hline
0  & 0.00089   & 1.3191   & 0.00028   & 2.493 & 300.545 & 80.355  & 150.203 & 18.952 & 0.4997 & 0.2358 \\
0.1       & 0.00091   & 1.3072   & 0.00029   & 2.465 & 297.876 & 80.593  & 146.984 & 18.872 & 0.4934 & 0.2341 \\
0.2       & 0.00094   & 1.2844   & 0.00031   & 2.441 & 295.010 & 81.658  & 144.138 & 18.869 & 0.4885 & 0.2310 \\
0.3       & 0.00097   & 1.2627   & 0.00033   & 2.408 & 291.902 & 82.572  & 140.528 & 18.818 & 0.4814 & 0.2279 \\
0.4       & 0.00281   & 0.5329   & 0.00027   & 2.531 & 208.917 & 122.608 & 122.893 & 19.891 & 0.5882 & 0.1622 \\
0.5       & 0.00104   & 1.2150   & 0.00036   & 2.367 & 285.233 & 84.626  & 135.755 & 18.859 & 0.4759 & 0.2228 \\
0.6       & 0.00285   & 0.5253   & 0.00028   & 2.511 & 207.847 & 122.908 & 121.491 & 19.884 & 0.5845 & 0.1617 \\
0.7       & 0.00282   & 0.5277   & 0.00028   & 2.503 & 209.051 & 123.322 & 121.277 & 19.824 & 0.5801 & 0.1607 \\
0.8       & 0.00112   & 1.1412   & 0.00041   & 2.313 & 284.037 & 90.723  & 130.889 & 18.988 & 0.4608 & 0.2093 \\
0.9       & 0.00271   & 0.5458   & 0.00029   & 2.496 & 213.458 & 123.669 & 122.077 & 19.741 & 0.5719 & 0.1596 \\
1         & 0.00263   & 0.5514   & 0.00030   & 2.488 & 219.303 & 126.339 & 123.241 & 19.715 & 0.5619 & 0.1560 \\ \hline
0.37      & 0.00279   & 0.5378   & 0.00026   & 2.537 & 209.275 & 122.221 & 123.529 & 19.946 & 0.5902 & 0.1632 \\ \hline
\end{tabular}
\end{table}

Estimators of parameters were obtained for different choices of tuning parameters. We then  computed the expected mean lifetime of the devices under the two doses of drug, both for the whole population ($E_{\omega=1}$ and $E_{\omega=2}$) and particularly for the mice that died without tumor ($E_{\omega=1}^1$ and $E_{\omega=2}^1$). We have also computed the  probability of failure due to cause $1$ (die without tumor) given failure, for both doses of drug ($P_{\omega=1}^1$ and $P_{\omega=2}^1$). 

We applied the procedure described in Section \ref{sec:sim_choice} to determine  the optimal tuning parameter for this data set, over a grid of $100$ points. The resulting optimal tuning parameter, $0.37$, and its corresponding estimators are presented in Table \ref{table:numdata2}.

Finally, we estimate the errors, as given in (\ref{eq:estimated_error}), for different tuning parameters $\beta$, and the corresponding results in Table \ref{table:BDCError}. The minimum is obtained for $\beta=0.8$, while $\beta=0.37$ also presents a lower estimated error, which is in concordance with the estimate obtained earlier.

\begin{table}[h]
\caption{Estimated errors for the BDC experiment \label{table:BDCError}}\setlength{\tabcolsep}{5pt}
\centering
\small 
\begin{tabular}{cccccccccccc}
\hline 
$\beta$ & $0$ & $0.1$ & $0.2$ & $0.3$ & $0.37$ & $0.4$ & $0.6$ & $0.7$ & $0.8$ & $0.9$ \\ 
\hline 
est. error & 0.1051 & 0.1049 & 0.1047 & 0.1044 &  0.1043 & 0.1051  & 0.1052 & 0.1050 & 0.1040 & 0.1048 \\ 
\hline 
\end{tabular} 
\end{table}


\section{Concluding Remarks and Future Work \label{sec:conclusion}}
In this paper, a robust divergence-based approach has been developed for the evaluation of one-shot devices with competing causes of failure, under the exponential distribution. The performance of the estimators as well as tests procedures based on them have been compared, through a simulation study and a numerical example, with these based on the classical maximum likelihood estimator.

For further study, we can consider developing results for other lifetime distributions, such as Weibull and gamma. While the exponential distribution has constant hazard rate, Weibull and gamma lifetime distributions presents a non-constant hazard and practically useful aging properties and may therefore provide a more practical model,  even though it will result in a much more complicated analysis. We are currently working on this problem and hope to report the findings in a future paper.

\bigskip

\noindent \textbf{Acknowledgments} This research was partially supported by Grant PGC2018-095194-B-I00 and Grant FPU16/03104 from Ministerio de Ciencia, Innovacion y Universidades (Spain). E. Castilla, N. Martin and L. Pardo are members of the Instituto de Matematica Interdisciplinar, Complutense University of Madrid.\\

\appendix
\section{Appendix \label{sec:app}}
\subsection{Proofs of Results}
\subsubsection{Proof of Theorem \ref{res:dkull}}
\begin{proof}
We have
\begin{align*}
\sum_{i=1}^{I}\frac{K_{i}}{K}d_{KL}(\widehat{\boldsymbol{p}}_{i},\boldsymbol{\pi}_{i}(\boldsymbol{\theta}))&=\frac{1}{K}\sum_{i=1}^{I}\left\{ {n_{i0}}\log \left( \dfrac{n_{i0}/K_i}{\pi _{i0}(\boldsymbol{\theta})}\right) +{n_{i1}}\log \left( \dfrac{n_{i1}/K_i}{\pi _{i1}(\boldsymbol{\theta})}\right) +{n_{i2}}\log \left( \dfrac{n_{i2}/K_i}{\pi _{i2}(\boldsymbol{\theta})}\right)\right\} \\
&=\frac{1}{K}\sum_{i=1}^{I}\left\{ {n_{i0}}\log \left( \dfrac{n_{i0}}{K_i}\right) +{n_{i1}}\log \left( \dfrac{n_{i1}}{K_i}\right) +{n_{i2}}\log \left( \dfrac{n_{i2}}{K_i}\right)\right\} \\
&-\frac{1}{K}\sum_{i=1}^{I}\left\{ {n_{i0}}\log \left( \pi_{i0}(\boldsymbol{\theta})\right) +{n_{i1}}\log \left( \pi_{i1}(\boldsymbol{\theta})\right) +{n_{i2}}\log \left( \pi_{i2}(\boldsymbol{\theta})\right)\right\}\\
&=c-\frac{1}{K}\log \left(\prod_{i=1}^I \pi_{i0}(\boldsymbol{\theta})^{n_{i0}}\pi_{i1}(\boldsymbol{\theta})^{n_{i1}}\pi_{i2}(\boldsymbol{\theta})^{n_{i2}}\right)\\
&=c-\frac{1}{K}\log \left(\mathcal{L}(\boldsymbol{\theta}|\boldsymbol{\delta},\boldsymbol{IT},\boldsymbol{x})\right),
\end{align*}
where $c=\frac{1}{K}\sum_{i=1}^{I}\sum_{r=0}^{2}\left\{ {n_{ir}}\log \left( \frac{n_{ir}}{K_i}\right) \right\}$ and it does not depend on the  parameter vector $\boldsymbol{\theta}$.
\end{proof}

\subsubsection{Proof of Theorem \ref{res:est_equations2}}
\begin{proof}
The estimating equations are given by
\begin{equation}\label{eq:weigh_eq}
\dfrac{\partial}{\partial\boldsymbol{\theta} } {}^*d_{\beta}^{W}(\boldsymbol{\theta})=\boldsymbol{0}_4,
\end{equation}
where $^*d_{\beta}^{Weighted}(\boldsymbol{\theta})$ is as given in (\ref{eq:dpd}). Equation (\ref{eq:weigh_eq}) is equivalent to
\begin{equation}
\frac{1}{\beta+1}\frac{\partial}{\partial\boldsymbol{\theta}}\sum_{i=1}^I\sum_{r=0}^2 K_i\pi_{ir}^{\beta+1}(\boldsymbol{\theta})-\frac{1}{\beta}\frac{\partial}{\partial\boldsymbol{\theta}}\sum_{i=1}^I\sum_{r=0}^2K_ip_{ir}\pi_{ir}^{\beta}(\boldsymbol{\theta})=\boldsymbol{0}_4;
\end{equation}
that is,
\begin{equation*}
\frac{1}{\beta+1}\sum_{i=1}^I\sum_{r=0}^2 K_i(\beta+1)\pi_{ir}^{\beta}(\boldsymbol{\theta})\frac{\partial\pi_{ir}(\boldsymbol{\theta})}{\partial\boldsymbol{\theta}}-\frac{1}{\beta}\sum_{i=1}^I\sum_{r=0}^2K_ip_{ir}\beta\pi_{ir}^{\beta-1}(\boldsymbol{\theta})\frac{\partial\pi_{ir}(\boldsymbol{\theta})}{\partial\boldsymbol{\theta}}=\boldsymbol{0}_4,
\end{equation*}
or, equivalently
\begin{equation*}
\sum_{i=1}^I\sum_{r=0}^2 K_i\pi_{ir}^{\beta-1}(\boldsymbol{\theta})\frac{\partial\pi_{ir}(\boldsymbol{\theta})}{\partial\boldsymbol{\theta}}[\pi_{ir}(\boldsymbol{\theta})-p_{ir}]=\boldsymbol{0}_4.
\end{equation*}
But,
\begin{align*}
\pi_{i0}(\boldsymbol{\theta})&=\exp(-(\lambda_{i1}+\lambda_{i2})IT_i),\\
\pi_{i1}(\boldsymbol{\theta})&=\frac{\lambda_{i1}}{\lambda_{i2}+\lambda_{2i}}(1-\exp(-(\lambda_{i1}+\lambda_{i2})IT_i))=\frac{\lambda_{i1}}{\lambda_{i2}+\lambda_{2i}}(1-\pi_{i0}(\boldsymbol{\theta})),\\
\pi_{i2}(\boldsymbol{\theta})&=\frac{\lambda_{i2}}{\lambda_{i1}+\lambda_{i2}}(1-\exp(-(\lambda_{i1}+\lambda_{i2})IT_i))=\frac{\lambda_{i2}}{\lambda_{i1}+\lambda_{i2}}(1-\pi_{i0}(\boldsymbol{\theta})),
\end{align*}
and so
$$
\frac{\partial \pi_{i0}(\boldsymbol{\theta})}{\partial \boldsymbol{\theta}}=-IT_i\pi_{i0}(\boldsymbol{\theta}) \frac{\partial}{\partial \boldsymbol{\theta}}[\lambda_{i1}+\lambda_{i2}]=-IT_i\pi_{i0}(\boldsymbol{\theta})(\lambda_{i1}/\theta_{10},\lambda_{i1}x_i,\lambda_{i2}/\theta_{20},\lambda_{i2}x_i)^T=-IT_i\pi_{i0}(\boldsymbol{\theta})\boldsymbol{l}_i,
$$

\begin{align*}
\dfrac{\partial \pi_{i1}(\boldsymbol{\theta})}{\partial \boldsymbol{\theta}}&=(1-\pi_{i0}(\boldsymbol{\theta}))\left[ \frac{\partial}{\partial \boldsymbol{\theta}}\frac{\lambda_{i1}}{\lambda_{i1}+\lambda_{i2}}\right]-\frac{\lambda_{i1}}{\lambda_{i1}+\lambda_{i2}}\dfrac{\partial \pi_{i0}(\boldsymbol{\theta})}{\partial \boldsymbol{\theta}},\\
\dfrac{\partial \pi_{i2}(\boldsymbol{\theta})}{\partial \boldsymbol{\theta}}&=(1-\pi_{i0}(\boldsymbol{\theta}))\left[ \frac{\partial}{\partial \boldsymbol{\theta}}\frac{\lambda_{i2}}{\lambda_{i1}+\lambda_{i2}}\right]-\frac{\lambda_{i2}}{\lambda_{i1}+\lambda_{i2}}\dfrac{\partial \pi_{i0}(\boldsymbol{\theta})}{\partial \boldsymbol{\theta}},
\end{align*}
where
$$
\left[ \frac{\partial}{\partial \boldsymbol{\theta}}\frac{\lambda_{i1}}{\lambda_{i1}+\lambda_{i2}}\right]=-\left[ \frac{\partial}{\partial \boldsymbol{\theta}}\frac{\lambda_{i2}}{\lambda_{i1}+\lambda_{i2}}\right]=\frac{\lambda_{i1}\lambda_{i2}}{(\lambda_{i1}+\lambda_{i2})^2}(1/\alpha_{10},x_i,-1/\alpha_{20},-x_i)^T.
$$
We then obtain the  desired result.
\end{proof}

\subsubsection{Proof of Theorem \ref{res:asymp}}

\begin{proof}
We have
\begin{align*}
\pi_{i0}(\boldsymbol{\theta})&=\exp(-(\lambda_{i1}+\lambda_{i2})IT_i),\\
\pi_{i1}(\boldsymbol{\theta})&=\frac{\lambda_{i1}}{\lambda_{i2}+\lambda_{2i}}(1-\exp(-(\lambda_{i1}+\lambda_{i2})IT_i))=\frac{\lambda_{i1}}{\lambda_{i1}+\lambda_{i2}}(1-\pi_{i0}(\boldsymbol{\theta})),\\
\pi_{i2}(\boldsymbol{\theta})&=\frac{\lambda_{i2}}{\lambda_{i1}+\lambda_{i2}}(1-\exp(-(\lambda_{i1}+\lambda_{i2})IT_i))=\frac{\lambda_{i2}}{\lambda_{i1}+\lambda_{i2}}(1-\pi_{i0}(\boldsymbol{\theta})).
\end{align*}
It is clear that
$$
\frac{\partial \pi_{i0}(\boldsymbol{\theta})}{\partial \boldsymbol{\theta}}=-IT_i\pi_{i0}(\boldsymbol{\theta}) \frac{\partial}{\partial \boldsymbol{\theta}}[\lambda_{i1}+\lambda_{i2}]=-IT_i\pi_{i0}(\boldsymbol{\theta})(\lambda_{i1}/\alpha_{10},\lambda_{i1}x_i,\lambda_{i2}/\alpha_{20},\lambda_{i2}x_i)^T.
$$
On the other hand $\pi_{i1}(\boldsymbol{\theta})=\frac{\lambda_{i1}}{\lambda_{i1}+\lambda_{i2}}(1-\pi_{i0}(\boldsymbol{\theta}))$ and $\pi_{i2}(\boldsymbol{\theta})=\frac{\lambda_{i2}}{\lambda_{i2}+\lambda_{i2}}(1-\pi_{i0}(\boldsymbol{\theta}))$, and so
\begin{align*}
\dfrac{\partial \pi_{i1}(\boldsymbol{\theta})}{\partial \boldsymbol{\theta}}&=(1-\pi_{i0}(\boldsymbol{\theta}))\left[ \frac{\partial}{\partial \boldsymbol{\theta}}\frac{\lambda_{i1}}{\lambda_{i1}+\lambda_{i2}}\right]-\frac{\lambda_{i1}}{\lambda_{i1}+\lambda_{i2}}\dfrac{\partial \pi_{i0}(\boldsymbol{\theta})}{\partial \boldsymbol{\theta}},\\
\dfrac{\partial \pi_{i2}(\boldsymbol{\theta})}{\partial \boldsymbol{\theta}}&=(1-\pi_{i0}(\boldsymbol{\theta}))\left[ \frac{\partial}{\partial \boldsymbol{\theta}}\frac{\lambda_{i2}}{\lambda_{i1}+\lambda_{i2}}\right]-\frac{\lambda_{i2}}{\lambda_{i1}+\lambda_{i2}}\dfrac{\partial \pi_{i0}(\boldsymbol{\theta})}{\partial \boldsymbol{\theta}}.
\end{align*}
Here, we have
$$
\left[ \frac{\partial}{\partial \boldsymbol{\theta}}\frac{\lambda_{i1}}{\lambda_{i1}+\lambda_{i2}}\right]=-\left[ \frac{\partial}{\partial \boldsymbol{\theta}}\frac{\lambda_{i2}}{\lambda_{i1}+\lambda_{i2}}\right]=\frac{\lambda_{i1}\lambda_{i2}}{(\lambda_{i1}+\lambda_{i2})^2}(1/\alpha_{10},x_i,-1/\alpha_{20},-x_i)^T.
$$
\end{proof}
\clearpage
\subsection{Power function of Wald-type tests \label{app:power}}
In many cases, the power function of the test procedure cannot be derived explicitly in small-sample situation. In the following result, we present a useful asymptotic result for approximating the power function of the  Wald-type test statistic given in equation (\ref{eq:WALDtest}).

\begin{theorem}
\label{th:wald2} Let $\boldsymbol{\theta}^{\ast }\notin \boldsymbol{\Theta }_{0}$ be the true value of the parameter such that $\widehat{\boldsymbol{\theta}}_{\beta }\underset{K\rightarrow \infty }{\overset{P}{\longrightarrow }}\boldsymbol{\theta}^{\ast }$, and let us denote 
\begin{equation*}
\ell _{\beta }\left( \boldsymbol{\theta}_{1},\boldsymbol{\theta}_{2}\right) =\boldsymbol{m}^{T}\left( \boldsymbol{\theta}_{1}\right) \left( \boldsymbol{M}^{T}\left( \boldsymbol{\theta}_{2}\right) \boldsymbol{\Sigma }_{\beta }\left( \boldsymbol{\theta}_{2}\right) \boldsymbol{M}\left( \boldsymbol{\theta}_{2}\right)\right) ^{-1}\boldsymbol{m}\left( \boldsymbol{\theta}_{1}\right) .
\end{equation*} 
We then have
\begin{equation*}
\sqrt{K}\left( \ell _{\beta }\left( \widehat{\boldsymbol{\theta}}_{1},\widehat{%
\boldsymbol{\theta}}_{2}\right) -\ell _{\beta }\left( \boldsymbol{\theta}^{\ast },%
\boldsymbol{\theta}^{\ast }\right) \right) \underset{K\rightarrow \infty }{%
\overset{\mathcal{L}}{\longrightarrow }}\mathcal{N}(0,\sigma _{W_{K},\beta
}^{2}\left( \boldsymbol{\theta}^{\ast })\right) ,
\end{equation*}%
where 
\begin{equation*}
\sigma _{W_{K},\beta }^{2}\left( \boldsymbol{\theta}^{\ast }\right) =\left. \frac{\partial \ell _{\beta }\left( \boldsymbol{\theta},\boldsymbol{\theta}^{\ast }\right) }{\partial \boldsymbol{\theta}^{T}}\right\vert _{\boldsymbol{\theta}=\boldsymbol{\theta}^{\ast }}\boldsymbol{\Sigma }_{\beta }\left( \boldsymbol{\theta}^{\ast }\right)\left. \frac{\partial \ell _{\beta }\left( \boldsymbol{\theta},\boldsymbol{\theta}^{\ast }\right) }{\partial \boldsymbol{\theta}}\right\vert _{\boldsymbol{\theta}=%
\boldsymbol{\theta}^{\ast }}.
\end{equation*}
\end{theorem}
\begin{proof}
Under the assumption that 
\begin{equation*}
\widehat{\boldsymbol{\theta}}_{\beta }\underset{K\rightarrow \infty }{\overset{P}{%
\longrightarrow }}\boldsymbol{\theta}^{\ast },
\end{equation*}%
the asymptotic distribution of $\ell _{\beta }\left( \widehat{\boldsymbol{\theta}}_{1},\widehat{\boldsymbol{\theta}}_{2}\right) $ coincides with the asymptotic distribution of $\ell _{\beta }\left( \widehat{\boldsymbol{\theta}}_{1},\boldsymbol{\theta}^{\ast }\right) .$ A first-order Taylor expansion of $\ell_{\beta }\left( \widehat{\boldsymbol{\theta}}_{\beta },\boldsymbol{\theta}\right) $ at $\widehat{\boldsymbol{\theta}}_{\beta }$ around, $\boldsymbol{\theta}^{\ast }$, yields

\begin{align*}
\left( \ell _{\beta }\left( \widehat{\boldsymbol{\theta}}_{\beta },\boldsymbol{\theta}^{\ast }\right) -\ell _{\beta }\left( \boldsymbol{\theta}^{\ast },\boldsymbol{\theta}^{\ast }\right) \right) &=\left. \frac{\partial \ell _{\beta }\left( \boldsymbol{\theta},\boldsymbol{\theta}^{\ast }\right) }{\partial \boldsymbol{\theta}^{T}}\right\vert _{\boldsymbol{\theta}=\boldsymbol{\theta}^{\ast }}\left( \widehat{\boldsymbol{\theta}}_{\beta }-\boldsymbol{\theta}^{\ast }\right)+o_{p}(K^{-1/2}).
\end{align*}%
Now, the result follows readily since 
\begin{equation*}
\sqrt{K}\left( \widehat{\boldsymbol{\theta}}_{\beta }-\boldsymbol{\theta}^{\ast}\right) \underset{K\rightarrow \infty }{\overset{\mathcal{L}}{\longrightarrow }}\mathcal{N}\left( \boldsymbol{0}_{4},\boldsymbol{\Sigma }_{\beta }\left( \boldsymbol{\theta}^{\ast }\right) \right) .
\end{equation*} 
\end{proof}

\begin{remark}
Based on Theorem \ref{th:wald2}, an approximation of the power function of the Wald-type test statistic in (\ref{eq:reject}) at $\boldsymbol{\theta}^*$ can be provided as follows:
\begin{align*}
\pi _{W,K}\left( \boldsymbol{\theta}^{\ast }\right)  &=\Pr \left( W_{K}\left( 
\widehat{\boldsymbol{\theta}}_{\beta }\right) >\chi _{r,\alpha }^{2}\right) =\Pr \left( K\left( \ell _{\beta }\left( \widehat{\boldsymbol{\theta}}_{\beta },\boldsymbol{\theta}^{\ast }\right) -\ell _{\beta }\left( \boldsymbol{\theta}^{\ast },\boldsymbol{\theta}^{\ast }\right) \right) >\chi _{r,\alpha }^{2}-K\ell _{\beta}\left( \boldsymbol{\theta}^{\ast },\boldsymbol{\theta}^{\ast }\right) \right) \\
& =\Pr \left( \frac{\sqrt{K}\left( \ell _{\beta }\left( \widehat{\boldsymbol{\theta}}_{\beta },\boldsymbol{\theta}^{\ast }\right) -\ell _{\beta }\left( \boldsymbol{\theta}^{\ast },\boldsymbol{\theta}^{\ast }\right) \right) }{\sigma _{W_{K},\beta}\left( \boldsymbol{\theta}^{\ast }\right) }  >\frac{1}{\sigma _{W_{K},\beta}\left( \boldsymbol{\theta}^{\ast }\right) }\left( \frac{\chi _{r,\alpha }^{2}}{\sqrt{K}}-\sqrt{K}\ell _{\beta }\left( \boldsymbol{\theta}^{\ast },\boldsymbol{\theta}^{\ast }\right) \right) \right) \\
& =1-\Phi _{K}\left( \frac{1}{\sigma _{W_{K},\beta }\left( \boldsymbol{\theta}%
^{\ast }\right) }\left( \frac{\chi _{r,\alpha }^{2}}{\sqrt{K}}-\sqrt{K}\ell
_{\beta }\left( \boldsymbol{\theta}^{\ast },\boldsymbol{\theta}^{\ast }\right) \right)
\right)
\end{align*}%
for a sequence of distributions functions $\Phi _{K}\left( x\right) $ tending uniformly to the standard normal distribution $\Phi \left( x\right) $. It is clear that 
\begin{equation*}
\lim_{K\rightarrow \infty }\pi _{W,K}\left( \boldsymbol{\theta}^{\ast }\right) =1,
\end{equation*}%
i.e., the Wald-type test statistic is consistent in the  sense of Fraser.
\end{remark}

The above approximation of the power function of the Wald-type test statistic can be used to obtain the sample size $K$ necessary in order to attain a prefixed power $\pi _{W,K}\left( \boldsymbol{\theta}^{\ast }\right) =\pi_{0}$. To do so,  it is necessary to solve the equation 
\begin{equation*}
\pi _{0}=1-\Phi _{K}\left( \frac{1}{\sigma _{W_{K},\beta }\left( \boldsymbol{%
a}^{\ast }\right) }\left( \frac{\chi _{r,\alpha }^{2}}{\sqrt{K}}-\sqrt{K}%
\ell _{\beta }\left( \boldsymbol{\theta}^{\ast },\boldsymbol{\theta}^{\ast }\right)
\right) \right) .
\end{equation*}%
The solution, in $K$, of the above equation yields  $\widehat{K}_{\beta }=\left[ \widehat{K}_{\beta }^{\ast }\right] +1$, where
\begin{equation*}
\widehat{K}_{\beta }^{\ast }=\frac{\widehat{A}_{\beta }+\widehat{B}_{\beta }+\sqrt{\widehat{A}_{\beta }(\widehat{A}_{\beta }+2\widehat{B}_{\beta })}}{2\ell _{\beta }^{2}\left( \boldsymbol{\theta}^{\ast },\boldsymbol{\theta}^{\ast
}\right) },
\end{equation*}%
with  \ $\widehat{A}_{\beta }=\sigma _{W_{K},\beta }^{2}\left( \boldsymbol{\theta}^{\ast}\right) \left( \Phi ^{-1}\left( 1-\pi _{0}\right) \right) ^{2}$ and \ $\widehat{B}_{\beta }=2\ell _{\beta }\left( \boldsymbol{\theta}^{\ast },\boldsymbol{\theta}^{\ast }\right) \chi _{r,\alpha }^{2}$.




\begin{thebibliography}{99}

\bibitem {bala2016a} Balakrishnan, N., So. H.,  and Ling, M. H. (2016a). A Bayesian approach for one-shot device testing with exponential lifetimes under competing risks. \emph{IEEE Transactions on Reliability}, \textbf{65}(1), 469--485.

\bibitem {bala2016b} Balakrishnan, N., So. H.,  and Ling, M. H. (2016b). EM algorithm for one-shot device testing with competing risks under Weibull distribution. \emph{IEEE Transactions on Reliability}, \textbf{65}(2), 973--991.

\bibitem{balacas2019a} Balakrishnan, N., Castilla, E., Martin N. and Pardo, L. (2019a). Robust estimators and test-statistics for one-shot device testing under the exponential distribution. \textit{IEEE Transactions on Information Theory}, \textbf{65}(5), 3080--3096.

\bibitem{balacas2019b} Balakrishnan, N., Castilla, E., Martin N. and Pardo, L. (2019b). Robust estimators for one-shot device testing data under gamma lifetime model with an application to a tumor toxicological data. \textit{Metrika}, \textbf{82}(8),  991--1019.

\bibitem{balacas2019c} Balakrishnan, N., Castilla, E., Martin N. and Pardo, L. (2019c). Robust inference for one-shot device testing data under Weibull lifetime model. \textit{IEEE Transactions on Reliability},  DOI: 10.1109/TR.2019.2954385.

\bibitem{balacas2020} Balakrishnan, N., Castilla, E., Martin N. and Pardo, L. (2020).  Robust inference for one-shot device testing data under exponential lifetime model with multiple stresses. \textit{Under revision}

\bibitem{crowder2001} Crowder, M. J. (2001). \textit{Classical Competing Risks}. Chapman and Hall/CRC, Press, London.

\bibitem {ghosh2013} Ghosh, A. and Basu, A. (2013). Robust estimation for independent non-homogeneous observations using density power divergence with applications to linear regression. \emph{Electronic Journal of Statistics}, \textbf{7}, 2420--2456.

\bibitem {ghosh2015} Ghosh, A. and Basu, A. (2015). Robust estimation for non-homogeneous data and the selection of the optimal tuning parameter: The density power divergence approach. \textit{Journal of Applied Statistics}, \textbf{42}, 2056--2072.

\bibitem {hampel} Hampel, F. R., Ronchetti, E., Rousseeuw, P. J. and Stahel W. (1986). \textit{Robust Statistics: The Approach Based on Influence Functions}.  John Wiley \& Sons, New York.

\bibitem {Heilbron} Heilbron, D. C. (1994). Zero-altered and other regression models for count data with added zeros. \textit{Biometrical Journal}, \textbf{36}, 531--547.

\bibitem {hong} Hong, C. and Kim, Y. (2001). Automatic selection of the tuning parameter in the minimum density power divergence estimation. \textit{Journal of the Korean Statistical Society}, \textbf{30}, 453--465.

\bibitem{Lambert} Lambert, D. (1992). Zero-inflated Poisson regression, with an application to defects in manufacturing. \textit{Technometrics}, \textbf{34}(1), 1–-14.

\bibitem{lindsey1993} Lindsey, J. and  Ryan, L. (1993). A three-state multiplicative model for rodent tumorigenicity experiments. \textit{Journal of the Royal Statistical Society, Series C}, \textbf{42}, 283--300.

\bibitem{so2016}  So, H. (2006). \textit{Some Inferential Results for One-Shot Device Testing Data Analysis}. PhD thesis, McMaster University, Canada; http://hdl.handle.net/11375/19438.

\bibitem{warwick} Warwick, J. (2001). \textit{Selecting Tuning Parameters in Minimum Distance Estimators}. PhD thesis, The Open University, Milton Keynes, England. 

\bibitem{warwick} Warwick, J. and Jones, M. C. (2005). Choosing a robustness tuning parameter. \emph{Journal of Statistical Computation and Simulation}, \textbf{75}, 581--588.
\end{thebibliography}
\end{document}